\documentclass{llncs}
\pdfoutput=1

\usepackage{amsthm}
\usepackage{booktabs}
\usepackage{rotating}
\usepackage{color}
\usepackage{caption}
\usepackage{listings}
\usepackage{xspace}
\usepackage{hyperref}
\usepackage{wrapfig} % to have figures wrapped by text
\usepackage{amsfonts}
\usepackage{amsmath}
\usepackage{amssymb}
\usepackage{temporal}
\usepackage{hyperref}
\usepackage{stmaryrd}
\usepackage{enumerate}
\usepackage{times}

\newcommand{\stateset}{\expandafter\MakeUppercase\expandafter{\state}}

\newcommand{\trans}{\ensuremath{\tau}}

\newcommand{\statevar}{x}
\newcommand{\statevarset}{\expandafter\MakeUppercase\expandafter{\statevar}}

\newcommand{\state}{\overline{\statevar}}
\newcommand{\allvar}{\mathcal{\statevarset}}
\newcommand{\obsvar}{o}
\newcommand{\obsvarset}{\expandafter\MakeUppercase\expandafter{\obsvar}}
\newcommand{\obsval}{\overline{\obsvar}}
\newcommand{\inputvar}{y}
\newcommand{\inputvarset}{\expandafter\MakeUppercase\expandafter{\inputvar}}
\newcommand{\inputval}{\overline{\inputvar}}
\newcommand{\val}{v}
\newcommand{\valset}{\bbB}%\expandafter\MakeUppercase\expandafter{\val}}
\newcommand{\valsetof}[1]{\valset^{#1}}
\newcommand{\trace}{\pi}
\newcommand{\restrict}[1]{{\restriction_{#1}}}
\newcommand{\randomvar}{r}
\newcommand{\randomvarset}{\expandafter\MakeUppercase\expandafter{\randomvar}}

\newcommand{\controlvar}{c}
\newcommand{\controlvarset}{\expandafter\MakeUppercase\expandafter{\controlvar}}

\newcommand{\controlval}{\overline{\controlvar}}
\newcommand{\memoryvar}{z}
\newcommand{\memoryvarset}{\expandafter\MakeUppercase\expandafter{\memoryvar}}

\newcommand{\memoryval}{\overline{\memoryvar}}
\newcommand{\memtrans}{\mu}
\newcommand{\sched}{{\sf sched}}
\newcommand{\schedvar}{s}
\newcommand{\bound}{b}
\newcommand{\card}[1]{\left| \,{#1}\, \right|}
\newcommand{\spec}{\Phi}
\newcommand{\impl}{~\rightarrow~}
\newcommand{\observe}{o}
\newcommand{\observations}{\expandafter\MakeUppercase\expandafter{\observe}}
\newcommand{\annot}{\lambda}
\newcommand{\cost}{{\sf cost}}

\newcommand{\mU}{\mathcal{U}}
\newcommand{\bbB}{\mathbb{B}}
\newcommand{\bbN}{\mathbb{N}}

\newcommand{\automaton}{\mU}

\definecolor{darkgreen}{rgb}{0,0.5,0}
\definecolor{darkblue}{rgb}{0,0,.5}
\definecolor{mygray}{gray}{.3}

\lstdefinelanguage{my_c}{morekeywords={
                  if, else, while, do, for, continue, break},
                           keywordstyle=\bf\color{darkblue},
                           directivestyle=\color{darkblue},
                           commentstyle=\color{darkgreen},
                           sensitive=true,
                           morecomment=[l]{//},
                           morecomment=[s]{/*}{*/},
                           morestring=[b]',
                           morestring=[b]",
                           moredelim=*[l][\color{darkgreen}]\#
                          }[keywords,comments,strings,directives]
\newcommand{\mrk}[1]{{\bf\color{red}#1}}
\newcommand{\ctrl}{?}

\newcommand{\ags}{AGS\xspace}

\makeatletter
\DeclareFontFamily{OMX}{MnSymbolE}{}
\DeclareSymbolFont{MnLargeSymbols}{OMX}{MnSymbolE}{m}{n}
\SetSymbolFont{MnLargeSymbols}{bold}{OMX}{MnSymbolE}{b}{n}
\DeclareFontShape{OMX}{MnSymbolE}{m}{n}{
    <-6>  MnSymbolE5
   <6-7>  MnSymbolE6
   <7-8>  MnSymbolE7
   <8-9>  MnSymbolE8
   <9-10> MnSymbolE9
  <10-12> MnSymbolE10
  <12->   MnSymbolE12
}{}
\DeclareFontShape{OMX}{MnSymbolE}{b}{n}{
    <-6>  MnSymbolE-Bold5
   <6-7>  MnSymbolE-Bold6
   <7-8>  MnSymbolE-Bold7
   <8-9>  MnSymbolE-Bold8
   <9-10> MnSymbolE-Bold9
  <10-12> MnSymbolE-Bold10
  <12->   MnSymbolE-Bold12
}{}

\let\llangle\@undefined
\let\rrangle\@undefined
\DeclareMathDelimiter{\llangle}{\mathopen}%
                     {MnLargeSymbols}{'164}{MnLargeSymbols}{'164}
\DeclareMathDelimiter{\rrangle}{\mathclose}%
                     {MnLargeSymbols}{'171}{MnLargeSymbols}{'171}
\makeatother

\newcommand{\ccat}{\circ}

\newif \ifextended

\usepackage{thmtools, thm-restate}

\usepackage{makeidx}  % allows for indexgeneration
\begin{document}
% uncomment the following line to get extended version
\extendedtrue
\frontmatter          % for the preliminaries
\pagestyle{plain}
\mainmatter              % start of the contributions
\title{Assume-Guarantee 
Synthesis for Concurrent Reactive Programs with Partial Information%
% \thanks{This work was supported in part by the Austrian Science Fund (FWF)
% through the national research network RiSE (S11406-N23 and S11407-N23) and the
% project QUAINT (I774-N23), as well as by the European Commission through project
% STANCE (317753).}
}
\titlerunning{Assume-Guarantee Synthesis with Limited Observability}
% abbreviated title (for running head)
% also used for the TOC unless
% \toctitle is used
%
\author{Roderick Bloem\inst{1}, 					
       Krishnendu Chatterjee\inst{2},
       Swen Jacobs\inst{1,3},
       Robert K\"onighofer\inst{1}
      }

\authorrunning{Bloem et al.} % abbreviated author list (for running head)
\institute{$^1$ IAIK, Graz University of Technology, Austria\\
           $^2$ IST Austria (Institute of Science and Technology Austria)\\
           $^3$ Reactive Systems Group, Saarland University, Germany}

\maketitle              % typeset the title of the contribution

\begin{abstract}
Synthesis of program parts is very useful for concurrent systems. However, most synthesis approaches do not support common design tasks, like modifying a single process without having to re-synthesize or verify the whole system.
Assume-guarantee synthesis (AGS) provides robustness against 
modifications of system parts, but thus far has been limited to the perfect 
information setting.  This means that local variables 
cannot be hidden from other processes, which renders synthesis
results cumbersome or even impossible to realize. We resolve this shortcoming by defining AGS in a partial information 
setting. We analyze the complexity and decidability in different settings, showing that the problem has a high worst-case complexity and is undecidable in 
many interesting cases. Based on these observations, we present a pragmatic algorithm based on bounded 
synthesis, and demonstrate its practical applicability on 
several examples. 
\end{abstract}
%\sj{maybe emphasize once more that these are reactive programs}

%\keywords{Reactive Synthesis; Program Sketching; Partial Information; Bounded
%Synthesis; Distributed Synthesis; SMT Solving}
%
\section{Introduction}
\label{sec:intro}

% Motivation: (synthesis is especially valueable in concurrent settings)
Concurrent programs %(in both hardware and software) 
are notoriously hard to get
right, due to unexpected behavior emerging from the interaction of different 
processes. At the same time, concurrency aspects such as mutual exclusion or
deadlock freedom are easy to express declaratively. This makes
concurrent programs an ideal subject for automatic synthesis. Due to the 
prohibitive complexity of synthesis 
tasks~\cite{PnueliR89,PnueliR90,FinkbeinerS05}, the automated 
construction of \emph{entire} programs from high-level specifications such as LTL is
often unrealistic. 
%In some cases it is also undesirable: (1) for certain parts
%of a program it may be easier to write imperative code than a formal
%specification, and (2) synthesized code may fail to reach the quality of a
%manual implementation, e.g., in terms of code size, maintainability, or
%resource consumption.
More practical approaches are based on partially 
implemented programs that should be completed or refined 
automatically~\cite{FinkbeinerS05,FinkbeinerJ12,Solar-Lezama13}, or program repair, 
where suitable replacements need to be synthesized for faulty program 
parts~\cite{JobstmannSGB12}.
This paper focuses on such applications, where parts of the system are
already given.

% Problem: Assume-guarantee is good, but we need partial information
When several processes need to be synthesized or refined
simultaneously, a fundamental question arises:  What are the assumptions about
the behavior of other processes on which a particular process should rely? 
The classical synthesis approaches assume either completely adversarial or 
cooperative
behavior, which leads to problems in both cases: adversarial components may 
result in unrealizability of the system, while cooperative components may may 
rely on a specific form of cooperation, and 
therefore are not robust against even small changes in a single process. 
Assume-Guarantee Synthesis
(\ags)~\cite{ChatterjeeH07} uses a more reasonable assumption: processes are 
adversarial, but will not violate their own specification to obstruct others.
Therefore, a 
system constructed by \ags will still satisfy its overall specification if we 
replace or refine one of the processes, as long as the new process satisfies
its local specification. 
Furthermore, \ags leads to the desired solutions in cases where 
the classical notions (of cooperative or completely adversarial processes)
 do not, for example in the synthesis of mutual 
exclusion protocols~\cite{ChatterjeeH07} or fair-exchange protocols 
for digital contract signing~\cite{CR14}. 

A drawback of existing algorithms for \ags~\cite{ChatterjeeH07,CR14} is that 
they only work in a perfect information setting. This means that each
component can access and use the values of all variables of the other
processes. This is a major restriction, as most concurrent implementations 
rely on variables that are \emph{local} to one process, and should not be changed
or observed by the other process. 
%In particular, it was shown in~\cite{JMM11} 
%that the partial-information setting allows for more robust modeling of 
%fair-exchange protocols. 
While classical notions of synthesis have been 
considered in such partial information settings 
before~\cite{MadhusudanT01,FinkbeinerS05}, we provide the first solution for 
AGS with partial information.

% Our work
\ifextended
In this work, we extend the AGS approach for simultaneous synthesis of multiple
processes with partial information restrictions, and analyze complexity and
decidability of AGS for several different cases. Furthermore, we provide the
first implementation of AGS, integrated into a programming model that combines
the synthesis of concurrent reactive programs with ideas from program 
sketching.
Our framework allows for a combined imperative-declarative programming style,
with fine-grained, user-provided restrictions on the exchange of information
between processes.  Our prototype implementation also
supports optimization of the synthesized program with respect to user-defined 
preferences, for example a small number of shared variables.  We 
demonstrate the value of
our approach on a number of small programs and protocols.

\smallskip
\noindent
\textbf{Complexity and Decidability of \ags.} 
We use reductions of assume-guarantee synthesis problems to problems about 
games with three players to obtain a number of new complexity results. 
We distinguish the general case, where synthesized programs may contain 
additional variables, from the memoryless case, where no variables may be 
added. We provide new complexity results for these two cases in both the 
perfect and the partial information setting, and for specifications in 
different fragments of linear-time temporal logic (LTL).
We show undecidability for general AGS under partial information for all 
fragments we consider, in particular for basic safety properties. 
Table~\ref{tab:complexity} gives an overview of the complexity of 
\ags.

\smallskip
\noindent
\textbf{Algorithms for \ags.}
In light of the high complexity of many \ags problems, we propose a pragmatic 
approach, based on program sketching and synthesis with bounded resources. 
Inspired by the bounded synthesis approach~\cite{FinkbeinerS13}, we reduce undecidable \ags 
problems under partial information to a sequence of decidable
\ags problems with bounded memory.

To this end, we formalize how to do bounded synthesis based on a program 
sketch. Our synthesis algorithm 
uses a translation of the specification into universal co-B\"uchi tree 
automata (cf.~\cite{FinkbeinerS13}), and an encoding of the existence of a correct instantiation
of the sketch
into a satisfiability modulo theories (SMT) problem. We show that the 
approach can be extended to the \ags setting by generating a number of 
separate SMT problems, and searching for a solution of their conjunction.

\smallskip
\noindent
\textbf{Implementation and Evaluation.}
We have implemented our algorithm and provide an evaluation on a number of
examples,
including Peterson's mutual exclusion protocol, a P2P filesharing protocol, a
double buffering protocol, and synthesis of atomic sections in a
concurrent device driver. We give sketches of these protocols that leave open
some decisions that are essential for correctness, and show that our AGS
algorithms finds suitable solutions.  Our tool also supports the optimization
of the synthesized implementation with respect to different metrics like
the number of memory updates or the size of atomic sections.  Using this
feature, we synthesize implementations that are both correct and optimal in a
certain sense.
Furthermore, we demonstrate how the robustness of \ags solution allows us to 
refine parts of the synthesized program without starting synthesis from 
scratch.
\else
\smallskip \noindent {\bf Contributions.} 
In this work, we extend assume-guarantee synthesis to the synthesis of
processes with partial information. In particular,
\begin{enumerate}[i)] 
\item we analyze the complexity and decidability of \ags by reductions 
to games with three players. We distinguish synthesis problems 
based on informedness (perfect or partial) and resources (bounded or unbounded 
memory) of processes, and specifications in different fragments of linear-time temporal 
logic (LTL). 
\item In light of the high complexity of many \ags problems, we propose a 
pragmatic approach, based on partially implemented programs
and synthesis with bounded resources.
We extend the bounded synthesis approach~\cite{FinkbeinerS13} to enable 
synthesis from partially defined, non-deterministic programs, and to the \ags setting.
\item We provide the first implementation of \ags, integrated into a 
programming model that allows for a combined imperative-declarative 
programming style with fine-grained, user-provided restrictions on the 
exchange of information
between processes. 
To obtain efficient and simple code, our prototype also
supports optimization of the synthesized program with respect to some basic user-defined 
metrics.
%, for example a small number of shared variables. 
\item We demonstrate the value of our approach on a number of small programs 
and protocols, including Peterson's mutual exclusion protocol, 
\ifextended a P2P  filesharing protocol, \fi
a double buffering protocol, and synthesis of atomic sections in a
concurrent device driver.
We also demonstrate how the robustness of \ags solutions allows us to 
refine parts of the synthesized program without starting synthesis from 
scratch.
\end{enumerate}
\fi

\section{Motivating Example}
\label{sec:ex}
%\sj{may be too complicated. Simplify! (see POPL Review \#139A)}

We illustrate our approach using the running example of~\cite{ChatterjeeH07}, 
a version of Peterson's mutual exclusion protocol.  More details can be found in 
\ifextended
Section~\ref{sec:exp_pet}.
\else
Appendix~\ref{sec:exp_pet}.
\fi

\smallskip
\noindent
\textbf{Sketch.} 
We use the term \emph{sketch} for concurrent reactive programs with non-deterministic choices.  
Listing~\ref{lst:p1} shows a sketch for Peterson's protocol with processes P1 
and P2. Variable \texttt{flag}$i$ indicates that P$i$ wants to enter the critical section, and
\texttt{cr}$i$ that P$i$ is in the critical section. The 
first \texttt{while}-loop waits for permission to enter the critical section, 
the second loop models some local computation. Question marks denote non-deterministic choices, and we want to synthesize expressions that replace question marks such that P1 and P2 never visit the critical section simultaneously.

\smallskip
\noindent
\textbf{Specification.}
The desired properties of both processes are (1) that whenever a process wants 
to enter the critical section, it will eventually enter it (starvation freedom), 
and (2) that the two processes are never in the critical section simultaneously 
(mutual exclusion). In Linear Temporal Logic (LTL)\footnote{In case the reader 
is not familiar with LTL: $\always$ is a temporal operator meaning ``in all time 
steps''; likewise $\eventually$ means ``at some point in the future''.}, this 
corresponds to the specification 
$\varphi_i = \always(\neg \texttt{cr1} \vee \neg \texttt{cr2}) \wedge 
\always(\texttt{flag}i \rightarrow \eventually \texttt{cr}i)$, 
for $i \in \{1,2\}$.

\begin{figure}[tb]
\captionof{lstlisting}{Sketch of Peterson's mutual exclusion protocol. \texttt{F}=$\false$, \texttt{T}=$\true$.}
\label{lst:p1}
\vspace{-0.4cm}
\begin{minipage}{0.49\textwidth}
\begin{lstlisting}[firstnumber=0]
                           turn:=F; flag1:=F; flag2:=F;
cr1:=F; wait1:=F;
do { // Process P1:
  flag1:=T;
  turn:=T;
  while(/*E\mrk{\ctrl$_{1,1}$}E*/) {} //wait /*E\label{lst:p1:q1}E*/
  cr1:=T;
  cr1:=F; flag1:=F; wait1:=T;
  while(/*E\mrk{\ctrl$_{1,2}$}E*/) {} //local work
  wait1:=F;
} while(T)
\end{lstlisting}
\end{minipage}
\hspace{0.1cm}
\begin{minipage}{0.49\textwidth}
\begin{lstlisting}[firstnumber=20]

cr2:=F; wait2:=F;
do { // Process P2:
  flag2:=T;
  turn:=F;
  while(/*E\mrk{\ctrl$_{2,1}$}E*/) {} //wait
  cr2:=T; //read:=/*E\mrk{\ctrl$_{2,3}$}E*/ /*E\label{lst:p2:crit}E*/
  cr2:=F; flag2:=F; wait2:=T;
  while(/*E\mrk{\ctrl$_{2,2}$}E*/) {} //local work
  wait2:=F;
} while(T)
\end{lstlisting}
\end{minipage}
\end{figure}

\smallskip
\noindent
\textbf{Failure of classical approaches.} 
There are essentially two options for applying standard synthesis techniques. 
First, we may assume that both processes are cooperative, and synthesize all 
\mrk{\ctrl$_{i,j}$} simultaneously.  However, the resulting implementation of P2 
may only work for the computed implementation of P1, i.e., changing P1 may break 
P2.  For instance, the solution \mrk{\ctrl$_{1,1}$} = \texttt{turn \& flag2}, 
\mrk{\ctrl$_{2,1}$} = \texttt{!turn} and \mrk{\ctrl$_{i,2}$} = \texttt{F} 
satisfies the specification, but changing \mrk{\ctrl$_{1,2}$} in P1 to 
\texttt{T} will make P2 starve. Note that this is not just a hypothetical case; 
we got exactly this solution in our experiments \ifextended 
(Section~\ref{sec:exp_pet}). \else (Appendix~\ref{sec:exp_pet})\fi. As a second 
option, we may assume that the processes are adversarial, i.e., P2 must work for 
\emph{any} P1 and vice versa. However, under this assumption, the problem is 
unrealizable~\cite{ChatterjeeH07}. 

\smallskip
\noindent
\textbf{Success of Assume-Guarantee Synthesis (\ags)~\cite{ChatterjeeH07}.} \ags 
fixes this dilemma by requiring that P2 must work for \emph{any} realization of 
P1 that satisfies its local specification (and vice 
versa). An \ags solution for Listing~\ref{lst:p1} is \mrk{\ctrl$_{1,1}$} = 
\texttt{turn \& flag2}, \mrk{\ctrl$_{2,1}$} = \texttt{!turn \& flag2} and 
\mrk{\ctrl$_{i,2}$} = \texttt{F}.

\smallskip
\noindent
\textbf{Added advantage of \ags.}
If one process in an \ags solution is changed or extended, but still satisfies 
its original specification, then the other process is guaranteed to do so as 
well.  We illustrate this feature by extending P2 with a new variable named 
\texttt{read}. It is updated in a yet unknown way (expressed by 
\mrk{\ctrl$_{2,3}$}) whenever P2 enters the critical section in 
line~\ref{lst:p2:crit} of Listing~\ref{lst:p1}.  Assume we want to implement 
\mrk{\ctrl$_{2,3}$} such that \texttt{read} is $\true$ and $\false$ infinitely 
often.  We take the solution from the previous paragraph and synthesize 
\mrk{\ctrl$_{2,3}$} such that P2 satisfies $\varphi_2 \wedge  
(\always\eventually\neg \texttt{read}) \wedge 
(\always\eventually\texttt{read})$, where $\varphi_2$ is the original 
specification of P2.  The fact that the modified process still satisfies 
$\varphi_2$ implies that P1 will still satisfy its original specification. We also notice that 
modular refinement saves overall synthesis time:  our tool takes  $19+55 = 74$ 
seconds to synthesize an \ags solution and refine it in a second step to get the expected solution with $\mrk{\ctrl_{2,3}} = \neg \texttt{read}$; direct 
synthesis of the refined specification for both processes requires $263$ 
seconds.

% Execute:
%   python3.2 ./ag_bosy.py --mono ../benchmarks/ag/pet0read
% Time:
%   55 seconds
% Solution:
% (define-fun f2_3 ((x!1 (_ BitVec 1))) (_ BitVec 1)
% (ite (= x!1 #b1) #b0
% (ite (= x!1 #b0) #b1
% #b0)))

% Execute:
%   python3.2 ./ag_bosy.py --mono ../benchmarks/ag/pet0readOnce
% Time:
%   263 seconds
% Solution:
% (define-fun f2_2 () (_ BitVec 1)
% #b0)
% (define-fun f1_2 () (_ BitVec 1)
% #b0)
% (define-fun f2_1 ((x!1 (_ BitVec 2))) (_ BitVec 1)
% (ite (= x!1 #b11) #b0
% (ite (= x!1 #b01) #b1
% (ite (= x!1 #b10) #b0
% (ite (= x!1 #b00) #b0
% #b0)))))
% (define-fun f1_1 ((x!1 (_ BitVec 2))) (_ BitVec 1)
% (ite (= x!1 #b11) #b1
% (ite (= x!1 #b01) #b0
% (ite (= x!1 #b00) #b0
% (ite (= x!1 #b10) #b1
% #b0)))))
% (define-fun f2_3 ((x!1 (_ BitVec 1))) (_ BitVec 1)
% (ite (= x!1 #b1) #b0
% (ite (= x!1 #b0) #b1
% #b0)))

\smallskip
\noindent
\textbf{Drawbacks of the existing~\cite{ChatterjeeH07} \ags framework.}  While 
\ags provides important improvements over classical approaches, it may still 
produce solutions like $\mrk{\ctrl_{1,1}} = \texttt{turn} \wedge \neg 
\texttt{wait2}$ and $\mrk{\ctrl_{2,1}} = \neg \texttt{turn} \wedge \neg 
\texttt{wait1}$. However, $\texttt{wait2}$ is intended to be a \emph{local} 
variable of P2, and thus invisible for P1. Solutions may also utilize modeling 
artifacts such as program counters, because \ags has no way to restrict the 
information visible to other processes. As a workaround, \cite{ChatterjeeH07} 
allows the user to define candidate implementations for each 
\mrk{\ctrl}, and let the synthesis algorithm select one of the candidates.  
However, this way, a significant part of the problem needs to be solved by 
the user.

\smallskip
\noindent
\textbf{\ags with partial information.} 
Our approach resolves this shortcoming by allowing the declaration of local 
variables. The user can write $f_{1,1}(\texttt{turn,flag2})$ instead of 
\mrk{\ctrl$_{1,1}$} to express that the solution may only depend on 
\texttt{turn} and \texttt{flag2}. Including more variables of P$1$ does 
not make sense for this example, because their value is fixed at the call site.  
When setting $\mrk{\ctrl_{2,1}} = f_{1,2}(\texttt{turn,flag1})$ (and 
$\mrk{\ctrl_{i,2}} = f_{i,2}()$), we get the solution proposed by Peterson: 
$\mrk{\ctrl_{1,1}} = \texttt{turn} \wedge \texttt{flag2}$ and $\mrk{\ctrl_{2,1}} 
= \neg \texttt{turn} \wedge \texttt{flag1}$ (and $\mrk{\ctrl_{i,2}} = 
\texttt{F}$).  This is the only \ags solution with these dependency 
constraints.

% Execute:
%   python3 ./ag_bosy.py ../benchmarks/ag2/pet_wait
% Time:
%   21 seconds
% Solution:
% (define-fun f2_2 () (_ BitVec 1)
% #b1)
% (define-fun f1_2 () (_ BitVec 1)
% #b0)
% (define-fun f2_1 ((x!1 (_ BitVec 2))) (_ BitVec 1)
% (ite (= x!1 #b10) #b0
% (ite (= x!1 #b00) #b1
% (ite (= x!1 #b11) #b0
% (ite (= x!1 #b01) #b0
% #b0)))))
% (define-fun f1_1 ((x!1 (_ BitVec 2))) (_ BitVec 1)
% (ite (= x!1 #b10) #b1
% (ite (= x!1 #b00) #b0
% (ite (= x!1 #b01) #b0
% (ite (= x!1 #b11) #b0
% #b0)))))

\smallskip
\noindent
\textbf{\ags with additional memory and optimization.} Our approach can also 
introduce additional memory in form of new variables. As with existing 
variables, the user can specify which question mark may depend on the memory 
variables, and also which variables may be used to update the memory. For our 
example, this feature can be used to synthesize the entire synchronization from 
scratch, without using \texttt{turn}, \texttt{flag1}, and \texttt{flag2}. 
Suppose we remove \texttt{turn}, allow some memory \texttt{m} instead, and 
impose the following restrictions: \mrk{\ctrl$_{1,1}$}$=f_{1,1}(\texttt{flag2}, 
\texttt{m})$, \mrk{\ctrl$_{2,1}$}$=f_{2,1}(\texttt{flag1}, \texttt{m})$, 
\mrk{\ctrl$_{i,2}$} is an uncontrollable input (to avoid overly simplistic 
solutions), and \texttt{m} can only be updated depending on the program counter 
and the old memory content.  Our approach also supports cost functions over the 
result, and optimizes solutions iteratively.  For our example, the user can 
assign costs for each memory update in order to obtain a simple solution with 
few memory updates. In this setup, our approach produces the solution presented 
in Listing~\ref{lst:p1sol}.
\begin{figure}[tb]
\captionof{lstlisting}{Result for
Listing~\ref{lst:p1}: \texttt{turn} is replaced by memory \mrk{m} in a clever
way.}
\label{lst:p1sol}
\vspace{-0.4cm}
\begin{minipage}{0.49\textwidth}
\begin{lstlisting}[firstnumber=0]
                     flag1:=F; flag2:=F; /*E\mrk{m:=F}E*/;
cr1:=F; wait1:=F;
do { // Process P1:
  flag1:=T;
  while(/*E\mrk{!m}E*/) {} //wait
  cr1:=T;
  cr1:=F; flag1:=F; wait1:=T;
  while(input1()) //work/*E\label{lst:p1sol:in1}E*/
    /*E\mrk{m:=F}E*/; 
  wait1:=F; /*E\mrk{m:=F;}E*/
} while(T)
\end{lstlisting}
\end{minipage}
\hspace{0.1cm}
\begin{minipage}{0.49\textwidth}
\begin{lstlisting}[firstnumber=20]

cr2:=F; wait2:=F;
do { // Process P2:
  flag2:=T;
  while(/*E\mrk{m}E*/) {} //wait
  cr2:=T;
  cr2:=F; flag2:=F; wait2:=T;
  while(input2()) //work/*E\label{lst:p1sol:in2}E*/
    /*E\mrk{m:=T};E*/
  wait2:=F; /*E\mrk{m:=T;}E*/
} while(T)
\end{lstlisting}
\end{minipage}
\end{figure}
It is surprisingly simple: It requires only one bit of memory \texttt{m}, 
ignores both \texttt{flag}s (although we did not force it to), and updates 
\texttt{m} only twice\footnote{The memory \texttt{m} is updated whenever an 
input is read in line~\ref{lst:p1sol:in1} or~\ref{lst:p1sol:in2}; we copied the 
update into both branches to increase readability.}.  Our proof-of-concept 
implementation took only 74 seconds to find this solution.

\section{Definitions}
\label{sec:prelim}

In this section we first define processes, refinement, schedulers, and 
specifications. Then we consider different versions of the co-synthesis problem, 
depending on \emph{informedness} (partial or perfect), \emph{cooperation}
(cooperative, competitive, assume-guarantee), and \emph{resources} (bounded or 
unbounded) of the players.

\smallskip
\noindent
\textbf{Variables, valuations, traces.} Let $\statevarset$ be a finite 
set of binary variables. A \emph{valuation} on $\statevarset$ is a function 
$\val: \statevarset \rightarrow \bbB$ that assigns to each variable $\statevar 
\in \statevarset$ a value $\val(\statevar) \in \bbB$. We write 
$\valsetof{\statevarset}$ for 
the set of valuations on $\statevarset$, and $u \ccat v$ for the concatenation 
of valuations $u \in \valsetof{\statevarset}$ and $v \in \valsetof{
\statevarset'}$ to a valuation in $\valsetof{\statevarset \cup \statevarset'}$.
 A \emph{trace} on $\statevarset$ is 
an infinite sequence $(\val_0,\val_1, \ldots)$ of 
valuations on $\statevarset$. Given a valuation $\val \in \valsetof{
\statevarset}$ and a 
subset $\statevarset' \subseteq \statevarset$ of the variables, define $\val 
\restrict{\statevarset'}$ as the \emph{restriction} of $v$ to 
$\statevarset'$. Similarly, for a trace $\trace = (\val_0,\val_1, \ldots)$ on 
$\statevarset$, write $\trace \restrict{\statevarset'} = (\val_0 
\restrict{\statevarset'}, 
\val_1 \restrict{\statevarset'}, \ldots)$ for the restriction of $\trace$ 
to the variables $\statevarset'$. The restriction operator extends naturally 
to sets of valuations and traces.

\smallskip
\noindent
\textbf{Processes and refinement.}
We consider non-deterministic processes, where the non-determinism is modeled by 
variables that are not under the control of the process. We call these variables 
\emph{input}, but they may also be internal variables with non-deterministic 
updates. For $i \in \{1,2\}$, a \emph{process} $P_i =  
(\statevarset_i,\obsvarset_i,\inputvarset_i,\trans_i)$ consists of finite sets
\begin{itemize}
\item $\statevarset_i$ of modifiable state variables,
\item $\obsvarset_i \subseteq \statevarset_{3-i}$ of observable (but not modifiable) state variables,
\item $\inputvarset_i$ of input variables,
\end{itemize}
and a \emph{transition function} $\trans_i: 
\valsetof{\statevarset_i} \times \valsetof{\obsvarset_i} \times \valsetof{\inputvarset_i} \rightarrow 
\valsetof{\statevarset_i}$. The transition function 
maps a current valuation of state and input variables to the next 
valuation for the state variables. We write $\statevarset = 
\statevarset_1 \cup \statevarset_2$ for the set of state 
variables of both processes, and similarly $\inputvarset = \inputvarset_1 \cup 
\inputvarset_2$ for the input variables. Note that some variables may be shared 
by both processes. Variables that are not shared between processes will be called \emph{local} variables. 

We obtain a refinement of a process by resolving some of the 
non-determinism introduced by input variables, and possibly extending the 
sets of local state variables. Formally, let $\controlvarset_i \subseteq 
\inputvarset_i$ be a set of \emph{controllable variables}, let $\inputvarset
_i' =  \inputvarset_i \setminus \controlvarset_i$, and let $\statevarset_i' 
\supseteq \statevarset_i$ be 
an extended (finite) set of state variables, with $\statevarset_1' \cap \statevarset_2' = \statevarset_1 \cap \statevarset_2$.
Then a \emph{refinement} of process $P_i = (\statevarset_i,\obsvarset_i,
\inputvarset_i, \trans_i)$ \emph{with respect to $C_i$} is a process $P_i' = 
(\statevarset_i',\obsvarset_i, \inputvarset'_i, \trans_i')$ with a  
transition function $\trans_i': \valsetof{\statevarset_i'}  \times 
\valsetof{\obsvarset_i} \times \valsetof{\inputvarset'_i} \rightarrow 
\valsetof{\statevarset_i'}$ such that for all $\state \in \valsetof{\statevarset_i'}
, \obsval \in \valsetof{\obsvarset_i}, \inputval \in 
\valsetof{\inputvarset'_i}$ there exists $\controlval \in 
\valsetof{\controlvarset_i}$ with
\[ \trans_i'(\state,\obsval,\inputval)\restrict{\statevarset_i} = 
\trans_i(\state\restrict{\statevarset_i}, \obsval, \inputval \ccat \controlval). \]
%
%old def:
 %Then, let $f_i: 
%\valsetof{\statevarset_i \cup \memoryvarset_i} \times \valsetof{\obsvarset_i} 
%\times \valsetof{\inputvarset'_i} 
%\rightarrow 
%\valsetof{\controlvarset_i}$ be a function that 
%computes values for the controllable variables, and $\memtrans_i: 
%\valsetof{\statevarset_i \cup \memoryvarset_i} \times \valsetof{\obsvarset_i} 
%\times \valsetof{\inputvarset'_i} \rightarrow 
%\valsetof{\memoryvarset_i}$ a transition function for the memory variables 
%(that maps valuations of state, memory, and input variables 
%to new valuations of the memory variables).
%Then a \emph{refinement} of process $P_i = (\statevarset_i,\obsvarset_i,
%\inputvarset_i, \trans_i)$ \emph{with respect to $C_i$} is a process $P_i' = (\statevarset_i \cup 
%\memoryvarset_i,\obsvarset_i, \inputvarset'_i, \trans_i')$ such that\sj{this definition does not allow to add observable variables} the 
%transition
%function $\trans_i': \valsetof{\statevarset_i \cup \memoryvarset_i}  \times 
%\valsetof{\obsvarset_i} \times \valsetof{\inputvarset'_i} \rightarrow 
%\valsetof{\statevarset_i \cup \memoryvarset_i}$ is defined 
%(for every $(\state,\memoryval) \in \valsetof{\statevarset_i \cup \memoryvarset
%_i}, \obsval \in \valsetof{\obsvarset_i}, \inputval \in \valsetof{\inputvarset'_i}$) as 
%\[\trans_i'((\state,\memoryval),\obsval,\inputval) = (\trans_i(\state,\obsval,
%(\inputval,f_i(\state,\obsval,\inputval))),
%\memtrans_i(\state,\obsval,\inputval)).\]
%
We write $P_i' \preceq P_i$ to denote that $P_i'$ is a refinement of $P_i$. 

\smallskip
\noindent
\textbf{Important modeling aspects.}
Local variables are used to model partial information: all decisions of a process need to be independent of the variables that are local to the other process.
Furthermore, variables in $\statevarset_i' \setminus \statevarset_i$ are used to model additional memory that a process can use to store observed information.
We say a refinement is \emph{memoryless} if $\statevarset_i' = \statevarset_i$, 
and it is \emph{$\bound$-bounded} if $\card{\statevarset_i' \setminus 
\statevarset_i} \leq \bound$.  

\smallskip
\noindent
\textbf{Schedulers, executions.}
A \emph{scheduler} for processes $P_1$ and $P_2$ chooses at each computation 
step whether $P_1$ or $P_2$ can take a step to update its variables. Let 
$\allvar_1,\allvar_2$ be the sets of all variables (state, memory, input) of $
P_1$ and $P_2$, respectively, and let $\allvar = \allvar_1 \cup 
\allvar_2$. Let furthermore $V=\valsetof{\allvar}$ be the set of global 
valuations.
Then, the scheduler is a function $\sched: V^* \rightarrow \{1,2\}$ 
that maps a finite sequence of global valuations to a process index $i \in \{1
,2\}$. Scheduler $\sched$ is 
\emph{fair} if for all traces $(v_0,v_1,\ldots) \in V^\omega$ it assigns 
infinitely many turns to both $P_1$ and $P_2$, i.e., there are infinitely 
many $j\geq 0$ such that $\sched(\val_0,\ldots,\val_j)=1$, and infinitely 
many $k\geq 0$ such that $\sched(\val_0,\ldots,\val_k)=2$.

Given two processes $P_1,P_2$, a scheduler $\sched$, and a start valuation $
\val_0$, the set of possible \emph{executions} of the parallel composition 
$P_1 \parallel P_2 \parallel \sched$ is 
\[ 
\llbracket P_1 \parallel P_2 \parallel \sched,\val_0 \rrbracket = \left\{ 
(\val_0,\val_1,\ldots) \in V^\omega \left| 
\begin{array}{l} \forall j \geq 0.\ \sched(\val_0,\val_1,\ldots,\val_j)=i \\ 
\textrm{ and } \val_{j+1}\restrict{(\allvar \setminus \allvar_i)} = 
\val_{j}\restrict{(\allvar \setminus \allvar_i)} \\
\textrm{ and } \val_{j+1}\restrict{\allvar_i \setminus \inputvarset_i} \in 
\trans_i(\val_{j}\restrict{\allvar_i}) 
\end{array}\right.\right\}. \]
That is, at every turn the scheduler decides which of the processes makes 
a transition, and the state and memory variables are updated according to the 
transition function of that process.
Note that during turns of process $P_i$, the values of local
variables of the other process (in $\allvar \setminus \allvar_i$) remain unchanged.

\smallskip
\noindent
\textbf{Safety, GR(1), LTL.}
A specification $\spec$ is a set of traces on $\statevarset \cup \inputvarset$.
We consider $\omega$-regular specifications, in particular the following 
fragments of LTL:\footnote{For a definition of syntax and semantics of LTL, 
see e.g.~\cite{ModelCheckBook}.}
\begin{itemize}
\item \emph{safety properties} are of the form $\always B$, where $B$ is a 
Boolean formula over variables in $\statevarset \cup \inputvarset$,
defining a subset of valuations that are safe.
\item \emph{GR(1) properties} are of the form
$ \left( \bigwedge_i \always \eventually L_e^i\right) \rightarrow \left( 
\bigwedge_j \always \eventually L_s^j\right)$,
where the $L_e^i$ and $L_s^j$ are Boolean formulas over $\statevarset \cup 
\inputvarset$.
\item \emph{LTL properties} are given as arbitrary LTL formulas over $
\statevarset \cup \inputvarset$. They are a 
subset of the $\omega$-regular properties.
\end{itemize}

\noindent
\textbf{Co-Synthesis.} \label{sec:cosydef}
In all co-synthesis problems, the input to the problem is 
given as: two processes $P_1,P_2$ with $P_i=(\statevarset_i,\obsvarset_i,\inputvarset_i,
\trans_i)$, two sets $\controlvarset_1,\controlvarset_2$ of controllable 
variables with $\controlvarset_i \subseteq \inputvarset_i$, two 
specifications $\spec_1, \spec_2$, and 
a start valuation $\val_0 \in \valsetof{\statevarset \cup \inputvarset}$, 
where $\inputvarset = \inputvarset_1 \cup \inputvarset_2$.

\smallskip
\noindent
\textit{Cooperative co-synthesis.}
The \emph{cooperative co-synthesis problem} is defined as follows: do there 
exist two processes $P_1' \preceq P_1$ and $P_2' \preceq P_2$, and a 
valuation $\val_0'$ with $\val_0' \restrict{\statevarset \cup \inputvarset} = 
\val_0$, such that for all fair schedulers $\sched$ we have 
$$\llbracket P_1' \parallel P_2' \parallel \sched,\val_0' \rrbracket \restrict
{\statevarset \cup \inputvarset} \subseteq \spec_1 \land \spec_2?$$

\noindent
\textit{Competitive co-synthesis.}
The \emph{competitive co-synthesis problem} is defined as follows: do there 
exist two processes $P_1' \preceq P_1$ and $P_2' \preceq P_2$, and a 
valuation $\val_0'$ with $\val_0' \restrict{\statevarset \cup \inputvarset} = 
\val_0$, such that for all fair schedulers $\sched$ we have 
\begin{enumerate}[(i)]
\item $\llbracket P_1' \parallel P_2 \parallel \sched,\val_0' \rrbracket \restrict{\statevarset 
\cup \inputvarset} \subseteq \spec_1$, and 
\item $\llbracket P_1 \parallel P_2' 
\parallel \sched,\val_0' \rrbracket \restrict{\statevarset \cup \inputvarset} 
\subseteq \spec_2$?
\end{enumerate}

\noindent
\textit{Assume-guarantee synthesis.}
The \emph{assume-guarantee synthesis (\ags) problem} is defined as follows: do there exist two processes $P_1' \preceq P_1$ and $P_2' \preceq P_2$, and a valuation $\val_0'$ with $\val_0' \restrict{\statevarset \cup \inputvarset} = \val_0$, such that for all fair schedulers $\sched$ we have 
\begin{enumerate}[(i)]
\item $\llbracket P_1' \parallel P_2 \parallel \sched,\val_0' \rrbracket \restrict{\statevarset \cup \inputvarset} \subseteq \spec_2 \impl \spec_1$, 
\item $\llbracket P_1 \parallel P_2' \parallel \sched,\val_0' \rrbracket \restrict{\statevarset \cup \inputvarset} \subseteq \spec_1 \impl \spec_2$, and 
\item $\llbracket P_1' \parallel P_2' \parallel \sched,\val_0' \rrbracket \restrict{\statevarset \cup \inputvarset} \subseteq \spec_1 \land \spec_2$?
\end{enumerate}

\noindent
%\textit{Perfect and partial information synthesis.}
We refer the reader to \cite{ChatterjeeH07} for more intuition and a detailed discussion of AGS. 

\smallskip
\noindent
\textit{Informedness and boundedness.} A synthesis problem is under \emph{perfect information} if $\statevarset_i \cup \obsvarset_i = \statevarset$ for $i \in \{1,2\}$, and $\inputvarset_1 = \inputvarset_2$. That is, both processes have knowledge about all variables in the system.
Otherwise, it is under \emph{partial information}.
%\textit{Memoryless and bounded synthesis.}
A synthesis problem is \emph{memoryless} (or $\bound$-bounded) if we additionally require that $P_1',P_2'$ are memoryless (or $\bound$-bounded) refinements of $P_1,P_2$.

\smallskip
\noindent
\textit{Optimization criteria.}
Let $\mathcal{P}$ be the set of all processes. A cost function is a function
$\cost: \mathcal{P} \times \mathcal{P} \rightarrow \bbN$ that assigns a cost to
a tuple of processes.  In our approach, we will use cost functions to optimize synthesis
results.

\smallskip
\noindent
\textit{Note on robustness against modifications.}
Suppose $P_1',P_2'$ are the result of \ags on a given input, including 
specifications $\spec_1,\spec_2$. The properties of \ags allow us to replace 
one of the processes, say $P_2$:
if the replacement of $P_2'$ satisfies $\spec_2$, then the overall system will still be correct. If we furthermore ensure that conditions ii) and iii) of \ags are satisfied, then the resulting solution is again an \ags solution, i.e., we can go on and refine another process.

\smallskip
\noindent
\textit{Co-synthesis of more than $2$ processes.}
The definitions above naturally extend to programs with more than $2$ 
concurrent processes, cp.~\cite{CR14} for \ags with $3$ processes.

\section{Complexity and Decidability of \ags}

\ifextended

We analyze the complexity of AGS, based on a reduction to graph-based games.

\subsection{Game Graphs for Co-Synthesis}
All synthesis problems defined thus far can be reduced to problems about games played on graphs with three players.

\paragraph{Game graphs.} \hspace{-2.3mm}
A \emph{$3$-player game graph} $G=((S,E),(S_1,S_2,S_3))$ consists of a
directed graph $(S,E)$ with a finite set $S$ of states and a set $E \subseteq S \times S$ of edges, and a partition $(S_1,S_2,S_3)$ of the state space $S$ into three sets. The states in $S_i$ are player-$i$ states, for $i \in \{1,2,3\}$. For a state $s \in S$, we write $E(s) = \{t \in S | (s,t) \in E\}$ for the set of successor states of $s$. We assume that every state has at least one outgoing edge; i.e., $E(s)$ is nonempty for all states $s \in S$. 

Beginning from a start state, the three players move a token along the edges of the game graph. If the token is on a player-$i$ state $s \in S_i$, then player $i$ moves the token along one of the edges going out of $s$. The result is an infinite path in the game
graph; we refer to such infinite paths as plays. Formally, a \emph{play} is an infinite
sequence $(s_0,s_1,s_2,\ldots)$ of states such that $(s_k,s_{k+1}) \in E$ for all $k \geq 0$. We write $\Omega$ for the set of plays. 

\paragraph{Strategies.}
A strategy for a player is a recipe that specifies how to extend plays.
Formally, a \emph{strategy} $\sigma_i$ for player $i$ is a function $\sigma_i 
: S^* \cdot S_i \rightarrow S$ that, given a finite sequence of states 
(representing the history of the play so far) which ends in a player-$i$ 
state, chooses the next state. The strategy must choose an available 
successor state; i.e., for all $w \in S^*$ and $s \in S_i$, if $\sigma_i (w 
\cdot s) = t$, then $t \in E(s)$.
We write $\Sigma_i$ for the set of strategies for player $i$. 

Strategies in general require memory to remember some facts about the history 
of a play. An equivalent definition of strategies is as follows: Let $M$ be a 
set called \emph{memory}.  A strategy $\sigma=(f,\memtrans)$ 
consists of (1) a next-state function 
$f: S \times M \rightarrow  S$ that, given the memory and the current 
state, determines the successor state, and (2) a memory-update function 
$\memtrans : S \times M \rightarrow M$  that, given the memory and the current 
state, updates the memory. 

The strategy $\sigma=(f,\memtrans)$ is \emph{finite-memory} if the 
memory $M$ is finite. It is \emph{$b$-bounded} if $2^b \geq \card{M}$, and 
\emph{memoryless} if $M$ is a singleton set (i.e., $b=0$). Memoryless 
strategies do not depend on the history of a play, but only on the current 
state. A memoryless strategy for player $i$ can be specified as a function $
f_i : S_i \rightarrow S$ such that $f_i(s) \in E(s)$ for all $s \in 
S_i$. Given a start state $s_0 \in S$ and three strategies $\sigma_i \in 
\Sigma_i$, one for each of the three players $i \in \{1,2,3\}$, there is a 
unique play, denoted $\omega(s_0,\sigma_1,\sigma_2,\sigma_3) = 
(s_0,s_1,s_2,\ldots)$, such that for all $k \geq 0$, if $s_k \in
S_i$, then $\sigma_i(s_0,s_1,\ldots,s_k) = s_{k+1}$; this play is the outcome 
of the game starting at $s_0$ given the three strategies $\sigma_1, \sigma_2$, 
and $\sigma_3$.

In a partial information setting, players may not be able to make decisions 
based on the full state of the game, but only with respect to the observed 
state. Formally, let $\observations$ be a set of observations. A \emph{
partial information strategy} with respect to an observation function $
\observe: S \rightarrow \observations$ is a strategy $\sigma$ with $\sigma(s_0
,s_1,\ldots,s_k)=\sigma(s_0',s_1',\ldots,s_k')$ whenever $\observe(s_i)=
\observe(s_i')$ for all $i$.

\paragraph{Winning.}
An objective $\Psi\subseteq \Omega$ is a set of plays; i.e., $\Psi \subseteq \Omega$. The following notation is derived from ATL~\cite{AlurHK02}. For an objective $\Psi$, the set of \emph{winning states} for player $1$ in the game graph $G$ is $\llangle 1 \rrangle_G (\Psi) = \{s\ {\in}\ S \mid \exists\ \sigma_1\ {\in}\ \Sigma_1.\ \forall\ \sigma_2\ {\in}\ \Sigma_2.\ \forall\ \\ \sigma_3\ {\in}\ \Sigma_3.\ \omega(s,\sigma_1,\sigma_2,\sigma_3)\ {\in}\ \Psi\}$; a witness strategy $\sigma_1$ for player $1$ for the existential quantifier is referred to as a \emph{winning strategy}.
The winning sets $\llangle 2 \rrangle_G (\Psi)$ and $\llangle 3\rrangle_G (\Psi)$ for players $2$ and $3$ are defined analogously. The set of winning states
for the team consisting of player $1$ and player $2$, playing against player $3$, is
$\llangle1,2\rrangle_G (\Psi) = \{s\ {\in}\ S \mid \exists\ \sigma_1\ {\in}\ \Sigma_1.\ \exists\ \sigma_2\ {\in}\ \Sigma_2.\ \forall\ \sigma_3\ {\in}\ \Sigma_3.\ \omega(s,\sigma_1,\sigma_2,\sigma_3)\ {\in}\ \Psi\}$. The
winning sets $\llangle I\rrangle_G (\Psi)$ for other teams $I \subseteq \{1,2,3\}$ are defined similarly. 

%Note that (full information?) finite-memory games with $\omega$-regular objectives are determined, i.e., the winning states for a team of players and a given specification $\spec$ are exactly those states where the opposing team does not have a winning strategy for $\neg \spec$.\cite{BL69}\sj{add formal version?}

\paragraph{Games based on processes and specifications.}
Given two processes $P_1, P_2$ with $P_i = (\statevarset_i,\obsvarset_i,
\inputvarset_i,
\trans_i)$ and respective sets of controllable variables $\controlvarset_i 
\subseteq \inputvarset_i$, we define the $3$-player game graph $G =
((S,E),(S_1,S_2,S_3))$ as follows: let $S = V \times \{1,2,3\}$; let $S_i = V 
\times \{i\}$ for
$i \in \{1,2,3\}$; and let $E$ contain (1) all edges of the form $((v,3),(u,i
))$ for $i \in \{1,2\}$, $v \in V$ and $u \restrict{\statevarset \cup 
\controlvarset} = v \restrict{\statevarset \cup \controlvarset}$, and (2) all 
edges
of the form $((v,i),(u,3))$ for $i \in \{1,2\}$ and $u \restrict \statevarset_
i \in \trans_i(v \restrict{\statevarset_i}, v\restrict{\obsvarset_i}, v 
\restrict{\inputvarset_i})$ and 
$u \restrict{\allvar \setminus (\statevarset_i \cup \controlvarset_i)} = v 
\restrict{\allvar \setminus (\statevarset_i \cup \controlvarset_i)}$.
In other words, player $1$ represents process $P_1$,
player $2$ represents process $P_2$, and player $3$ represents the 
environment, including the scheduler. Given a
play of the form $\omega = ((v_0,3),(v_0',i_0 ),(v_1,3),(v_1',i_1),(v_2,3),
\ldots)$, where $i_j \in \{1,2\}$ for all $j \geq 0$, we write $[\omega]_{1,2
}$ for the sequence of valuations $(v_0',v_1',v_2',\ldots)$ in $\omega$
(ignoring the intermediate valuations at player-3 states).\footnote{Note that 
$v_j$ differs from $v_j'$ only in the valuation of input variables, and $v_j'$
 differs from $v_{j+1}$ only in the valuation of variables in 
$\statevarset_{i_j} \cup \controlvarset_{i_j}$, controlled by process $i_j$.} 

A given specification $\spec \subseteq
 V^\omega$ defines the objective $[[\spec]] = \left\{\omega \in \Omega | [
\omega]_{1,2} \in \spec\right\}$. In this way, the specifications
$\spec_1$ and $\spec_2$ for the processes $P_1$ and $P_2$ provide the 
objectives $\Psi_1 = [[\spec_1]]$ and $\Psi_2 = [[\spec_2]]$ for players $1$ 
and $2$, respectively. The objective for player 3 (the environment) is the 
fairness objective $\Psi_3 = \fair$ that both $S_1$ and $S_2$ are visited
infinitely often; i.e., $\fair$ contains all plays $(s_0,s_1,s_2,\ldots) \in 
\Omega$ 
such that $s_j \in S_1$ for infinitely many $j \geq 0$, and $s_k \in S_2$ for 
infinitely many $k \geq 0$.

\paragraph{Game solutions to co-synthesis problems~\cite{ChatterjeeH07}.}
Based on a game graph as defined above, the cooperative co-synthesis problem for $P_1, P_2, \controlvarset_1, \controlvarset_2$, a start valuation $v_0$ and specifications $\spec_1, \spec_2$ is equivalent to finding a winning strategy for the team of players $1$ and $2$ from start valuation $v_0$, and the objective $\Psi = [[\fair \rightarrow \spec_1 \land \spec_2]]$. The corresponding competitive co-synthesis problem is equivalent to finding separate strategies for players $i \in \{1,2\}$ for this game graph from start valuation $v_0$, and the respective objective $\Psi_i = [[\fair \rightarrow \spec_i]]$. 

For the \ags problem for $P_1, P_2, \controlvarset_1, \controlvarset_2$, a 
start valuation $v_0$ and specifications $\spec_1, \spec_2$, consider the 
following:
\begin{enumerate}
\item let $U_i = \llangle i \rrangle_G (\fair \rightarrow \Psi_i)$ be the 
winning states for process $i$, based on a fair scheduler,
\item let $F_i = \llangle i, 3 \rrangle_{G\restrict{U_i}} (\fair \land \Psi_i 
\land \neg \Psi_{3-i})$ be the set of states where the team of players $i$ 
and $3$ can win the game and force the other player to lose the game, and 
\item let $W = \llangle 1, 2 \rrangle_{G\restrict{S \setminus \left(F_1 \cup F
_2\right)}} (\fair \rightarrow \left( \Psi_1 \land \Psi_2)\right)$ be the set 
of states where both players $1$ and $2$ can win the game, but not force the 
other to lose it, based on a fair scheduler.
\end{enumerate}
Then the \ags problem is equivalent to finding strategies $\sigma_1, \sigma_2$
 for players $1$ and $2$, respectively, such that:
\begin{enumerate}
\item player $i$ wins the game with objective $(\fair \land \Psi_{3-i}) 
\rightarrow \Psi_i$ from all states in $U_i$: $\forall \sigma_{3-i}'.\ \forall
 \sigma_3.\ \forall s \in U_i.\ \omega(s,\sigma_i,\sigma_{3-i}',\sigma_3) \in 
((\fair \land \Psi_{3-i}) \rightarrow \Psi_i)$,
\item the team of players $1$ and $2$ wins the game with objective $\fair 
\rightarrow (\Psi_1 \land \Psi_2)$ from states $W \setminus (U_1 \cup U_2)$, 
and 
\item $v_0 \in W$.
\end{enumerate}
Formally, solving the AGS problem reduces to solving games with secure 
equilibria~\cite{ChatterjeeH07}.

\subsection{Complexity Results}

Table~\ref{tab:complexity} gives an overview of the complexity of 
\ags.
The complexity results are with respect to the size of the input, where the 
input consists of the game graph given explicitly, and the specification 
formula (i.e., the size of the input is the 
size of the explicit game graph and the length of the formula). 

\begin{table}
\centering
\begin{tabular}{|l|c|c|c|c|}
\hline
& \multicolumn{2}{|c|}{Memoryless} & \multicolumn{2}{|c|}{General} \\
\hline
& Perfect & Partial & Perfect & Partial \\
\hline
Safety & P & NP-C & P & Undec\\
GR(1) & NP-C & NP-C & P & Undec\\
LTL & PSPACE-C & PSPACE-C & 2EXP-C & Undec\\
\hline
\end{tabular}
\caption{Complexity of Assume-Guarantee Synthesis}
\label{tab:complexity}
\end{table}

\else

We give an an overview of the complexity of \ags.
The complexity results are with respect to the size of the input, where the 
input consists of the state transition system, and the specification 
formula (i.e., the size of the input is the 
size of the explicit state transition system and the length of the formula).

\begin{theorem}\label{thm:complexity}
The complexity of \ags is given in the following table:

\centering
\begin{tabular}{|l|c|c|c|c|}
\hline
& \multicolumn{2}{|c|}{Bounded Memory} & \multicolumn{2}{|c|}{~Unbounded Memory~} \\
\hline
& Perfect Inf. & Partial Inf. & Perfect Inf. & ~Partial Inf.~ \\
\hline
~Safety~ & P & NP-C & P & Undec\\
~GR(1) & NP-C & NP-C & P & Undec\\
~LTL & ~PSPACE-C~ & ~PSPACE-C~ & ~2EXP-C~ & Undec\\
\hline
\end{tabular}

\end{theorem}

\noindent
Formal definitions of three-player games, and proof ideas for these complexity results by reduction to three-player games can be found in Appendix~\ref{sec:complexity}.

\fi
Note that the complexity classes for memoryless \ags  are the same as 
for \ags with bounded memory --- the case of bounded memory reduces to the 
memoryless case, by considering a game that is larger by a constant factor: 
the given bound.

Also note that if we consider the results in the order given by the 
columns of the table, they form a non-monotonic pattern:
(1)~For safety objectives the complexity increases and then decreases (from 
PTIME to NP-complete to PTIME again);
(2)~for GR(1) objectives it remains NP-complete and finally decreases to 
PTIME; and
(3)~for LTL it remains PSPACE-complete and then increases to 2
EXPTIME-complete.

\ifextended
We will explain these results in the following.

\subsubsection{Memoryless AGS, Perfect Information.}

The following Theorem justifies the results in the first column of
\ifextended
Table~\ref{tab:complexity}.
\else
the table in Theorem~\ref{thm:complexity}.
\fi

\begin{theorem}
\label{thm:compl-memless-perf}
The complexity of memoryless AGS under perfect information is 
\begin{enumerate}[i)]
\item \label{thm:compl-memless-perf-safe} polynomial for safety properties,
\item \label{thm:compl-memless-perf-gr1} NP-complete for GR(1) properties, and
\item \label{thm:compl-memless-perf-ltl} PSPACE-complete for LTL properties.
\end{enumerate}
\end{theorem}

\begin{proof}
We present the proof of the three items below.

Item \ref{thm:compl-memless-perf-safe}:
It was shown in~\cite{ChatterjeeH07} that AGS solutions can be obtained from the solutions of games with 
secure equilibria. It follows from the results of~\cite{CHJ06} that for games with safety objectives,
the solution for secure equilibria reduces to solving games with safety and reachability objectives 
for which memoryless strategies suffice (i.e., memoryless strategies are as powerful as arbitrary 
strategies for safety objectives). It also follows from~\cite{CHJ06} that for safety objectives, games with 
secure equilibria can be solved in polynomial time.

Item \ref{thm:compl-memless-perf-gr1}: 
It follows from the results of~\cite{FHW80} that even in a graph (not a game) the question whether there exists a memoryless
strategy to visit two distinct states infinitely often is NP-hard (a reduction from directed subgraph 
homeomorphism).
Since visiting two distinct states infinitely often is a conjunction of two B\"uchi objectives, 
which is a special case of GR(1) objectives, the lower bound follows.
For the NP upper bound, the witness memoryless strategy can be guessed, and once a memoryless
strategy is fixed, we have a graph, and the polynomial-time verification procedure is the 
polynomial-time algorithm for model checking graphs with GR(1) objectives~\cite{PPS06}.

Item \ref{thm:compl-memless-perf-ltl}: 
In the special case of a game graph where every player-1 state has exactly one outgoing edge,
the memoryless AGS problem is an LTL model checking problem, and thus the lower bound of
LTL model checking~\cite{ModelCheckBook} implies PSPACE-hardness.
For the upper bound, we guess a memoryless strategy (as in Item \ref{thm:compl-memless-perf-gr1}),
and the verification problem is an LTL model checking question.
Since LTL model checking is in PSPACE~\cite{ModelCheckBook} and NPSPACE=PSPACE (by Savitch's theorem)~\cite{Savitch70,ComplexityBook},
we obtain the desired result.
\end{proof}

\subsubsection{Memoryless AGS, Partial Information.}
The following Theorem justifies the results in the second column of 
\ifextended
Table~\ref{tab:complexity}.
\else
the table in Theorem~\ref{thm:complexity}.
\fi

\begin{theorem}
\label{thm:compl-memless-part}
The complexity of memoryless AGS under partial information is 
\begin{enumerate}[i)]
\item \label{thm:compl-memless-part-safe} NP-complete for safety properties,
\item \label{thm:compl-memless-part-gr1} NP-complete for GR(1) properties, and
\item \label{thm:compl-memless-part-ltl} PSPACE-complete for LTL properties.
\end{enumerate}
\end{theorem}

\begin{proof}
We present the proof of the three items below.

Item~\ref{thm:compl-memless-part-safe}:
The lower bound result was established in~\cite{CKS13}. 
For the upper bound, again the witness is a memoryless strategy. Given
the fixed strategy, we have a graph problem with safety and reachability
objectives that can be solved in polynomial time (for the polynomial-time verification). 

Item~\ref{thm:compl-memless-part-gr1}: 
The lower bound follows from Theorem~\ref{thm:compl-memless-perf}, Item~\ref{thm:compl-memless-perf-gr1}; and the upper bound is similar as well.

Item~\ref{thm:compl-memless-part-ltl}: 
Similar to Theorem~\ref{thm:compl-memless-perf}, Item~\ref{thm:compl-memless-perf-ltl}.
\end{proof}

\subsubsection{General AGS, Perfect Information.}
The following Theorem justifies the results in the third column of 
\ifextended
Table~\ref{tab:complexity}.
\else
the table in Theorem~\ref{thm:complexity}.
\fi

\begin{theorem}
\label{thm:compl-general-perf}
The complexity of general AGS under perfect information is 
\begin{enumerate}[i)]
\item \label{thm:compl-general-perf-safe} polynomial for safety properties,
\item \label{thm:compl-general-perf-gr1} polynomial for GR(1) properties, and
\item \label{thm:compl-general-perf-ltl} 2EXP-complete for LTL properties.
\end{enumerate}
\end{theorem}

\begin{proof}
We present the proof of the three items below.

Item~\ref{thm:compl-general-perf-safe}: For AGS under perfect information 
and safety objectives, the memoryless and the general problem coincide 
(as mentioned in Theorem~\ref{thm:compl-memless-perf}, Item~\ref{thm:compl-memless-perf-safe}).
The result follows from Theorem~\ref{thm:compl-memless-perf}, Item~\ref{thm:compl-memless-perf-safe}.

Item~\ref{thm:compl-general-perf-gr1}: 
It follows from the results of~\cite{ChatterjeeH07,CHJ06} that solving AGS for perfect-information games 
requires solving games with implication conditions. Since games with implication of GR(1) objectives
can be solved in polynomial time~\cite{GameBook}, the desired result follows. 

Item~\ref{thm:compl-general-perf-ltl}: 
The lower bound follows from standard LTL synthesis~\cite{PnueliR89}.
For the upper bound, AGS for perfect-information games requires solving implication games,
and games with implication of LTL objectives can be solved in 2EXPTIME~\cite{PnueliR89}.
The desired result follows.
\end{proof}

\subsubsection{General AGS, Partial Information.} 
The following Theorem justifies the results in the fourth column of 
\ifextended
Table~\ref{tab:complexity}.
\else
the table in Theorem~\ref{thm:complexity}.
\fi

\begin{theorem}
\label{thm:undec-general-part}
General AGS under partial information is undecidable for safety properties.
\end{theorem}

\begin{proof}
It was shown in~\cite{Peterson79} that three-player partial-observation games are undecidable, and it was also shown that the undecidability result holds for safety objectives as well~\cite{CHOP13}.
\end{proof}

%\begin{corollary}
%General AGS under partial information is undecidable for GR(1) and LTL %properties.
%\end{corollary}

\else
\fi

\section{Algorithms for \ags}

Given the undecidability of \ags in general, and its high complexity for most 
other cases, we propose a pragmatic approach that divides the general 
synthesis problem into a sequence of synthesis problems with a bounded
amount of memory, and encodes the resulting problems into SMT formulas. Our 
encoding is inspired by the \emph{Bounded Synthesis} 
approach~\cite{FinkbeinerS13}, but supports synthesis from non-deterministic program sketches, 
as well as \ags problems. By iteratively deciding
whether there exists an implementation for an increasing bound on the number of 
memory variables, we obtain a semi-decision procedure for \ags with partial 
information.

We first define the procedure for cooperative co-synthesis problems, and then 
show how to extend it to \ags problems.

\subsection{SMT-based Co-Synthesis from Program Sketches}
\label{sec:SMT-encoding}

Consider a \emph{cooperative} co-synthesis problem with inputs $P_1$ and $P_2$, 
defines as 
$P_i=(\statevarset_i,\obsvarset_i,\inputvarset_i,
\trans_i)$, two sets $\controlvarset_1,\controlvarset_2$ of controllable 
variables with $\controlvarset_i \subseteq \inputvarset_i$, a
specification $\spec_1 \land \spec_2$, and 
a start valuation $\val_0 \in \valsetof{\statevarset \cup \inputvarset}$, 
where $\inputvarset = \inputvarset_1 \cup \inputvarset_2$.

In the following, we describe a set of SMT constraints such that a model  
represents refinements $P_1'\preceq P_1,P_2'\preceq P_2$ such 
that for all fair schedulers $\sched$, we have
$\llbracket P_1' \parallel P_2' \parallel \sched, \val_0\rrbracket \subseteq 
\spec_1 \land \spec_2$. Assume we are given a bound $b \in \bbN$, and let $
\memoryvarset_1, 
\memoryvarset_2$ be disjoint sets of additional memory variables with 
$\card{\memoryvarset_i}=b$ for $i \in \{1,2\}$.

\medskip
\noindent
\textbf{Constraints on given transition functions.}
In the expected way, the transition functions $\trans_1$ and $\trans_2$ are 
declared as functions $\trans_i: \valsetof{\statevarset_i} \times 
\valsetof{\obsvarset_i} \times \valsetof{\inputvarset_i} \rightarrow 
\valsetof{\statevarset_i}$, and directly encoded into SMT constraints by 
stating 
$ \trans_i(\state,\obsval,\inputval)=\state' $
for every $\state \in \valsetof{\statevarset_i},\obsval \in 
\valsetof{\obsvarset_i},\inputval \in \valsetof{\inputvarset_i}$, according 
to the given transition functions $\trans_1, \trans_2$.

\medskip
\noindent
\textbf{Constraints for interleaving semantics, fair scheduling.}
To obtain an encoding for interleaving semantics, we 
add a scheduling variable $\schedvar$ to both sets of inputs $\inputvarset_1$ 
and $\inputvarset_2$, and require that (i)~$\trans_1(\state,\obsval,\inputval) = \state$
whenever $\inputval(\schedvar) = \false$, and (ii)
$\trans_2(\state,\obsval,\inputval) = \state$
whenever $\inputval(\schedvar) = \true$.
Fairness of the scheduler can then be encoded as the LTL formula
$\always \eventually \schedvar \land \always \eventually \neg \schedvar$, 
abbreviated $\fair$ in the following.

\medskip
\noindent
\textbf{Constraints on resulting strategy.}
Let $\statevarset_i' = \statevarset_i \cup \memoryvarset_i$ be the extended 
state set, and $\inputvarset_i' = \inputvarset_i \setminus \controlvarset_i$ 
the reduced set of input variables of process $P_i'$. Then the resulting 
strategy of $P_i'$ is represented by functions $\memtrans_i: \valsetof{
\statevarset_i'} \times \valsetof{\obsvarset_i} \times \valsetof{
\inputvarset_i'} \rightarrow \valsetof{\memoryvarset_i}$ to update the
memory variables, and $f_i : \valsetof{\statevarset_i'} \times \valsetof{
\obsvarset_i} \times \valsetof{\inputvarset_i'} \rightarrow \valsetof{
\controlvarset_i}$ to resolve the non-determinism for controllable variables.
Functions $f_i$ and $\memtrans_i$ for $i \in \{1,2\}$ are constrained 
indirectly using constraints on an auxiliary annotation function 
that will ensure that the resulting strategy satisfies the specification $
\spec = (\fair \rightarrow \spec_1 \land \spec_2)$.
To obtain these constraints, first transform $\spec$ into a universal co-B\"
uchi automaton $\automaton_\spec = (Q,q_0,\Delta,F)$, where
\begin{itemize}
\item $Q$ is a set of states and $q_0 \in Q$ is the initial state,
\item $\Delta \subseteq Q \times Q$ is a set of transitions, labeled with 
valuations $\val \in \valsetof{\statevarset_1 \cup \statevarset_2 \cup 
\inputvarset_1 \cup \inputvarset_2}$, and
\item $F \subseteq Q$ is a set of rejecting states.
\end{itemize}
The automaton is such that it rejects a trace if it violates $\spec$, i.e., 
if rejecting states are visited infinitely often. 
Accordingly, it accepts a concurrent program $(P_1 \parallel P_2 \parallel 
\sched, \val_0)$ if no trace in $\llbracket P_1 \parallel P_2 \parallel \sched
, \val_0\rrbracket$ violates $\spec$. See~\cite{FinkbeinerS13} for more
background.

Let $\statevarset' = \statevarset_1' \cup \statevarset_2'$. We constrain 
functions $f_i$ and $\memtrans_i$ with respect to an 
additional annotation function $\annot: Q \times \valsetof{\statevarset'} 
\rightarrow \bbN \cup \{\bot\}$. In the following, let $\trans_i'(\state \ccat
\memoryval,\obsval,\inputval)$ denote the combined update function for the 
original state variables and additional memory variables, explicitly written 
as $$\trans_i(\state \ccat \memoryval,\obsval,
\inputval \ccat f_i(\state,\memoryval,\obsval,\inputval)) \ccat 
\memtrans_i(\state \ccat \memoryval,\obsval,\inputval).$$
Similar to the original bounded synthesis 
encoding~\cite{FinkbeinerS13}, we require that
\[ \annot(q_0,\val_0\restrict{\statevarset'}) \in \bbN. \]
If (1) $(q,(\state_1,\state_2))$ is a composed state with $\annot(q,(\state_1,
\state_2)) \in \bbN$, (2) $\inputval_1 \in \valsetof{\inputvarset_1},
\inputval_2 \in \valsetof{\inputvarset_1}$ are inputs and $q'\in Q$ is a 
state of the automaton such that there is a transition $(q,q') \in \Delta$ 
that is labeled with $(\inputval_1,\inputval_2)$, and (3) $q'$ is a non-
rejecting state of $\automaton_\spec$, then we require
\[ \annot(q',(\trans_1'(\state_1,\obsval_1,\inputval_1),\trans_2'(\state_2,
\obsval_2,\inputval_2))) \geq \annot(q,(\state_1,\state_2)), \]
where values of $\obsval_1,\obsval_2$ are determined by values of $\state_2$ 
and $\state_1$, respectively (and the subset of states of one process which 
is observable by the other process).

\noindent
Finally, if conditions (1) and (2) above hold, and $q'$ is rejecting 
in $\automaton_\spec$, we require
\[ \annot(q',(\trans_1'(\state_1,\obsval_1,\inputval_1),\trans_2'(\state_2,
\obsval_2,\inputval_2))) > \annot(q,(\state_1,\state_2)). \]
Intuitively, these constraints ensure that in no execution starting from 
$(q_0,\val_0)$, the automaton will visit rejecting states infinitely often.
Finkbeiner and Schewe~\cite{FinkbeinerS13} have shown that these constraints 
are satisfiable if 
and only if there exist implementations of $P_1,P_2$ with state variables $
\statevarset_1,\statevarset_2$ that 
satisfy $\spec$. With our additional constraints on the original $\trans_1, 
\trans_2$ and the integration of the $f_i$ and $\memtrans_i$ as new 
uninterpreted functions, they are satisfiable if there exist $\bound$-bounded 
refinements 
of $P_1, P_2$ (based on $\controlvarset_1, \controlvarset_2$) that satisfy $
\spec$. An SMT solver can then be used to find interpretations of the $f_i$ 
and $\memtrans_i$, as well as the auxiliary annotation functions that witness 
correctness of the refinement.

\medskip
\noindent
\textbf{Correctness.} 
The proposed algorithm for bounded synthesis 
from program sketches is correct and will eventually find a solution if it 
exists:

\begin{restatable}{proposition}{PropOne}
Any model of the SMT constraints will represent a refinement of the program 
sketches such that their composition satisfies the specification.
\end{restatable}
\ifextended
\begin{proof}
From our definitions of refinement and of the transition functions $\trans_i'$,
it is obvious that a model will represent a refinement of the given program 
sketches.

Furthermore, by correctness of the annotation approach from 
bounded synthesis~\cite{FinkbeinerS13}, any transition function that 
satisfies the constraints will satisfy the specification (and the combination 
of $\trans_i'$ is in particular a transition function).
\end{proof}
\else
\fi
\begin{restatable}{proposition}{PropTwo}
There exists a model of the SMT constraints if there exist $\bound$-bounded 
refinements $P_1' \preceq P_1, P_2' \preceq P_2$ that satisfy the specification.
\end{restatable}

\ifextended
\begin{proof}
Suppose such $P_1', P_2'$ exist. By the definition of refinement, we have 
that for all $\state \in \valsetof{\statevarset_i'}
, \obsval \in \valsetof{\obsvarset_i}, \inputval \in 
\valsetof{\inputvarset'_i}$ there exists $\controlval \in 
\valsetof{\controlvarset_i}$ with
\[ \trans_i'(\state,\obsval,\inputval)\restrict{\statevarset_i} = 
\trans_i(\state\restrict{\statevarset_i}, \obsval, \inputval \ccat \controlval
). \]

The control valuations $\controlval$ for different valuations $\state,\obsval,
\inputval$ of the other variables give us a 
model of the function $f_i$ that computes the controllable variables. In a 
similar way, the computation of memory valuations for different valuations of 
the other variables gives us a model of the function $\memtrans_i$.
\end{proof}
\else
\noindent
Proof ideas for correctness of the algorithm can be found in Appendix~\ref{correctness-appendix}.
\fi

\medskip
\noindent
\textbf{Optimization of solutions.} 
Let $\cost: \mathcal{P} \times \mathcal{P} \rightarrow \bbN$ be a user-defined
const function. We can synthesize an implementation $P_1',P_2'\in \mathcal{P}$ 
with \emph{maximal cost} $b$ by adding the constraint $\cost(P_1',P_2') \leq b$
(and a definition of the $\cost$ function), and we can \emph{optimize} the 
solution by searching for implementations with incrementally smaller cost. For
instance, a cost function could count the number of memory updates in order to 
optimize solutions for simplicity.

\subsection{SMT-based \ags}

Based on the encoding from Section~\ref{sec:SMT-encoding}, this section 
presents an extension that solves the \ags problem. Recall that the inputs to 
\ags are two program sketches 
$P_1,P_2$ with 
$P_i=(\statevarset_i,\obsvarset_i,\inputvarset_i,
\trans_i)$, two sets $\controlvarset_1,\controlvarset_2$ of controllable 
variables with $\controlvarset_i \subseteq \inputvarset_i$, two specifications
$\spec_1, \spec_2$, and 
a start valuation $\val_0 \in \valsetof{\statevarset \cup \inputvarset}$, 
where $\inputvarset = \inputvarset_1 \cup \inputvarset_2$. The goal is to 
obtain refinements $P_1' \preceq P_1$ and $P_2' \preceq P_2$ such that:
\begin{enumerate}[(i)]
\item $\llbracket P_1' \parallel P_2 \parallel \sched, \val_0 \rrbracket 
\subseteq (\fair \land \spec_2 \rightarrow \spec_1)$
\item $\llbracket P_1 \parallel P_2' \parallel \sched, \val_0 \rrbracket 
\subseteq (\fair \land \spec_1 \rightarrow \spec_2)$
\item $\llbracket P_1' \parallel P_2' \parallel \sched, \val_0 \rrbracket 
\subseteq (\fair \rightarrow \spec_1 \land \spec_2)$.
\end{enumerate}
Using the approach presented above, we can encode each of the three items 
into a 
separate set of SMT constraints, using the same function symbols and variable 
identifiers in all three problems. In more detail, this means that we
\begin{enumerate}
\item encode (i), where we ask for a model of $f_1$ and 
$\memtrans_1$ such that $P_1'$ with $\trans_1'$ and $P_2$ with the given 
$\trans_2$ satisfy the first property,
\item encode (ii), where we ask for a model of $f_2$ and 
$\memtrans_2$ such that $P_1$ with the given $\trans_1$ and $P_2'$ with $
\trans_2'$ satisfy the second property, and 
\item encode (iii), where we ask for models of $f_i$ and $\memtrans_i
$ for $i \in \{1,2\}$ such that $P_1'$ and $P_2'$ with $\trans_1'$ and 
$\trans_2'$ satisfy the third property.
\end{enumerate}
Then, a solution for the conjunction of all of these constraints must be such 
that the resulting refinements of $P_1$ and $P_2$ satisfy all three 
properties simultaneously, and are thus a solution to the \ags problem. 
Moreover, a solution to the SMT problem exists if and only if there exists 
a solution to the \ags problem.

\subsection{Extensions}
While not covered by the definition of \ags in Section~\ref{sec:prelim}, we can easily extend our algorithm to the following cases:
\begin{enumerate}
\item If we allow the sets $\memoryvarset_1, \memoryvarset_2$ to be non-disjoint, then the synthesis algorithm can refine processes also by \emph{adding shared variables}.
\item Also, our algorithms can easily be adapted to \ags with \emph{more than 2 processes}, as defined in~\cite{CR14}.
\end{enumerate}

\ifextended
\section{Implementation}
\label{sec:impl}

We have implemented our AGS approach with partial information as an extension 
to BoSY, the bounded synthesis backend of parameterized synthesis tool \textsc{Party}~\cite{KhalimovJB13a}. It uses 
LTL3BA~\cite{BabiakKRS12} to transform specifications into automata, and 
Z3~\cite{MouraB08} as SMT solver.  Our extension is available 
for download\footnote{\url{
http://www.student.tugraz.at/robert.koenighofer/tacas15_AG.zip}}.

\textbf{Input.}
Our tool takes three input files: a program sketch and one specification file 
for each process. The sketch is defined directly in SMT-LIBv2~\cite{Barret10} 
format, the specifications are given in LTL, using the Acacia~\cite{BohyBFJR12} 
syntax.

The sketch defines data types for the state space $\valsetof{\statevarset}$, the 
uncontrollable input space $\valsetof{\inputvarset'}$, the controllable input 
space $\valsetof{\controlvarset}$, and (optional) memory $\valsetof{ 
\memoryvarset}$, along with the initial valuation $v_0$ of all variables.% 
\footnote{Wlog., we assume that memory variables are initialized to a default 
value, e.g. $\false$ for Boolean variables.} For each Boolean signal $s$ that 
appears in the specification, the sketch defines a labeling function $s: 
\valsetof{\statevarset} \times \valsetof{\inputvarset}  \rightarrow \mathbb{B}$, 
which is by default just the value of a state or input variable. For signals of 
the specification that are not directly available as state- or input variables, 
the labeling function needs to be explicitly defined.

In our experiments, we mostly use bitvectors of appropriate length to define 
state space, inputs, and memory.  However, our tool also supports the definition 
of user-defined data types such as tuples of enumeration types, which may be 
more convenient for other applications.  Our tool uses a special integer 
constant $M$ to refer to the number of memory variables per process, and 
increases $M$ until a solution is found. All memory variables are global by 
default. Partial information is modeled by restricting the set of variables on 
which the functions that control the strategy or update the memory can depend. 
This fine-grained definition of partial information increases the flexibility of 
our tool.

%In the simplest case, $\valsetof{\memoryvarset}$ can be
%defined as bitvector of size $M$.  But $\valsetof{\memoryvarset}$ can also be
%composed of several sub-memories where only some grow with $M$, or a tuple of
%$M$ sub-memories of the same size, for instance. 
By default, the sketch defines the (global) transition function $\trans$ as the
parallel composition of the transition functions $\tau_1$ and $\tau_2$ of the
processes, but sometimes defining the combined transition function directly 
is easier. 
Finally, the sketch declares the functions that should be synthesized:
two control functions $f_i$, and two memory update functions $\memtrans_i$.
The user can specify each of these functions compositionally, with multiple 
sub-functions that control disjoint subsets of $\controlvarset_i$ or $
\memoryvarset_i$, respectively. For each sub-function, 
observable variables can be defined individually, allowing for a very 
fine-grained use of partial information.

\textbf{Optimization of solutions.}
In order to facilitate the optimization of solutions, the user can assert that
some arbitrarily computed cost has to be lower than some
special constant $\mathsf{Opt}$ in the sketch file.  Our tool will find the
minimal value of $\mathsf{Opt}$, within a user-defined interval, such
that the problem is still realizable.  At the moment, this search is implemented
in a straightforward way: $\mathsf{Opt}$ is decreased (increased) by 1 as long
the problem is realizable. More sophisticated search strategies like binary
search, %restarting with the costs of the last solution, 
learning from failed
attempts, %by analyzing the unsatisfiable core, 
or using incremental solving are
possible but not yet implemented.  With a solver for MAX-SMT problems, this
search could also be entrusted to the solver.

\textbf{Other applications.}
Our tool can also complete sketches with cooperative co-syn\-thesis (see
Section~\ref{sec:cosydef}).  Furthermore, we can use our tool as a model 
checker for solutions by completely defining the functions $f_i$ and 
$\memtrans_i$ in the sketch instead of just declaring them.

The fact that our tool takes as input an SMT-LIBv2 formulation of the 
synthesis problem makes it very flexible:  By defining the transition 
relation appropriately, it can also be used for synthesis of entire systems 
from scratch, or synthesis of atomic sections in concurrent programs. 
Furthermore, additional requirements on the solution can be defined easily 
with additional SMT constraints.  The obvious downside is limited usability, 
since defining the program semantics in SMT-LIBv2 format is not always easy.  
In future work, we want to implement a front-end to define simple sketching 
problems in a subset of C in order to increase the usability.

\else
% in the conference version, we only have a short paragraph at
% the beginning of experiments.tex
\fi

\section{Experiments}
\label{sec:exp}
% \sj{need a better motivating example for AGS? where specification (naturally) changes over time? currently, ``selection of changes seems completely ad hoc''}
% \sj{include ``limits'' of our algorithm, i.e., examples where it times out or needs very long?}
% \sj{benefits of AGS should be pointed out/explained as goals of the evalutation 
% section}

\ifextended
% in the extended version, we hava a separate section describing our
% implementation (implementation.tex).
\else
% in the conference version, we only include a short paragraph here:
We implemented\footnote{Available at {\scriptsize\url{ 
http://www.student.tugraz.at/robert.koenighofer/tacas15_AG.zip}}.} our 
approach as an extension to \textsf{BoSY}, the bounded synthesis backend of
the parameterized synthesis tool \textsf{PARTY}~\cite{KhalimovJB13a}. The 
user defines 
the sketch in \textsf{SMT-LIB} format with a special naming scheme.  The 
specification is given in LTL.  The amount of memory is defined using a 
special 
integer constant $M$, which is increased until a solution is found. To 
optimize 
solutions, the user can assert that some arbitrarily computed cost must be 
lower 
than some constant $\mathsf{Opt}$.  Our tool will find the minimal value of 
$\mathsf{Opt}$ such that the problem is still 
realizable.  Our tool can also run cooperative co-synthesis as well as verify 
existing solutions. Appendix~\ref{sec:impl} contains more details.
\fi

All experiments in this section were performed on an ordinary 
notebook 
(Intel i5-3320M CPU@2.6 GHz, 8 GB RAM, 64-bit Linux), using one CPU core 
and a memory limit of 1 GB, which was never reached. 
\ifextended

\subsection{Peterson's Mutual Exclusion Protocol}
\label{sec:exp_pet}
% \sj{here and in motivation, search space is small because of pre-defined busy-wait loops. can we ``encode some of the control decisions (e.g., conditional 
%   jump back to candidate labels) in your input to AGS and still get a 
%   reasonable result?''}

\noindent
This example has already been used as motivation in Section~\ref{sec:ex}.
In this section, we give additional insights, experiments and performance
measures.

In our model of this program, we use a bitvector of size 7 to represent states:
Each process has a program counter of 3 bits (assignments written in the same
line in Listing~\ref{lst:p1} are executed simultaneously), and one bit is used
to model the variable \texttt{turn}.  The current value of all other variables
can be computed from the respective program counter value. Hence, they are
modeled with state labels.

\textbf{\ags without partial information.}
When using \ags without any restrictions on variable dependencies, our tool 
takes 26 seconds to find a solution. However, it is too complicated to be shown 
here, let alone understand it.  The question marks are implemented as what 
appears to be arbitrary functions over all variables, including program counter 
bits from the other process. This solution is overly complicated and thus 
clearly undesirable. This motivates our partial information extension to \ags.

\textbf{Cooperative co-synthesis.}
In our next experiment, we therefore restrict the observable information for 
resolving the question marks by setting
\mrk{\ctrl$_{1,1}$}$=f_{1,1}(\texttt{turn}, \texttt{flag2})$,
\mrk{\ctrl$_{2,1}$}$=f_{2,1}(\texttt{turn}, \texttt{flag1})$, and 
\mrk{\ctrl$_{i,2}$}$=f_{i,2}()$.  
Furthermore, we disable \ags (i.e., use cooperative co-synthesis) and do not 
allow extra memory. Our tool takes 7 seconds to find the solution shown in 
Listing~\ref{lst:p1solm}.
\begin{figure}[tb]
\captionof{lstlisting}{Synthesis result for the sketch in
Listing~\ref{lst:p1} without \ags.}
\label{lst:p1solm}
\vspace{-0.4cm}
\begin{minipage}{0.49\textwidth}
\begin{lstlisting}[firstnumber=0]
                    turn:=F; flag1:=F; flag2:=F;
cr1:=F; wait1:=F;
do { // Process P1:
  flag1:=T;
  turn:=T;
  while(/*E\mrk{turn \& flag2}E*/) {} //wait
  cr1:=T;
  cr1:=F; flag1:=F; wait1:=T;
  while(/*E\mrk{F}E*/) {} //local work /*E\label{lst:p1solm:c}E*/
  wait1:=F;
} while(T)
\end{lstlisting}
\end{minipage}
\hspace{0.1cm}
\begin{minipage}{0.49\textwidth}
\begin{lstlisting}[firstnumber=20]

cr2:=F; wait2:=F;
do { // Process P2:
  flag2:=T;
  turn:=F;
  while(/*E\mrk{!turn}E*/) {} //better: & flag1 /*E\label{lst:p1solm:p}E*/
  cr2:=T;
  cr2:=F; flag2:=F; wait2:=T;
  while(/*E\mrk{F}E*/) {} //local work
  wait2:=F;
} while(T)
\end{lstlisting}
\end{minipage}
% Execute:
%   python3 ./ag_bosy.py --mono ../benchmarks/ag2/pet0
% Time:
%   7 seconds
% Solution:
% (define-fun f2_2 () (_ BitVec 1)
% #b0)
% (define-fun f1_2 () (_ BitVec 1)
% #b0)
% (define-fun f1_1 ((x!1 (_ BitVec 2))) (_ BitVec 1)
% (ite (= x!1 #b01) #b0
% (ite (= x!1 #b11) #b1
% (ite (= x!1 #b10) #b0
% (ite (= x!1 #b00) #b0
% #b0)))))
% (define-fun f2_1 ((x!1 (_ BitVec 2))) (_ BitVec 1)
% (ite (= x!1 #b01) #b1
% (ite (= x!1 #b00) #b1
% (ite (= x!1 #b11) #b0
% (ite (= x!1 #b10) #b0
% #b1)))))
\end{figure}
The problem with this solution is that P2 relies on the concrete realization of
P1. If we would later modify the condition in line~\ref{lst:p1solm:c} (i.e.,
\mrk{\ctrl$_{1,2}$}) to $\true$, then P2 would starve while waiting for P1 to 
set \texttt{turn}. 

\textbf{\ags with partial information.}
AGS prevents such dependencies on the concrete realization of other processes, 
thereby making the solution robust against a posteriori changes of single 
processes.  Indeed, when running our tool \emph{with} AGS, we get \texttt{!turn 
\& flag1} in line~\ref{lst:p1solm:p}, which resolves the problem.  The execution 
time increases from 7 to 19 seconds, which is acceptable.
% Execute:
%   python3 ./ag_bosy.py ../benchmarks/ag2/pet0
% Time:
%   19 seconds
% Solution:
% (define-fun f2_2 () (_ BitVec 1)
% #b1)
% (define-fun f1_2 () (_ BitVec 1)
% #b0)
% (define-fun f2_1 ((x!1 (_ BitVec 2))) (_ BitVec 1)
% (ite (= x!1 #b11) #b0
% (ite (= x!1 #b01) #b1
% (ite (= x!1 #b10) #b0
% (ite (= x!1 #b00) #b0
% #b0)))))
% 
% (define-fun f1_1 ((x!1 (_ BitVec 2))) (_ BitVec 1)
% (ite (= x!1 #b11) #b1
% (ite (= x!1 #b10) #b0
% (ite (= x!1 #b01) #b0
% (ite (= x!1 #b00) #b0
% #b0)))))

\textbf{Introducing memory.}
So far, we assumed that the synchronization variables are already present. 
However, by introducing additional memory variables, our synthesis approach can 
also invent them. In our next experiment, we remove \texttt{turn}, allow some 
memory $m$ to be updated based on the program counter (of the currently 
scheduled process) and the old memory, and set 
\mrk{\ctrl$_{1,1}$}$=f_{1,1}(\texttt{m}, \texttt{flag2})$, and 
\mrk{\ctrl$_{2,1}$}$=f_{2,1}(\texttt{m}, \texttt{flag1})$.  We get the solution 
depicted in Listing~\ref{lst:p1solmem} within 19 seconds.
\begin{figure}[tb]
\captionof{lstlisting}{Synthesis result for the sketch in
Listing~\ref{lst:p1} with \ags and memory, but without
optimization for simplicity.}
\label{lst:p1solmem}
\vspace{-0.4cm}
\begin{minipage}{0.44\textwidth}
\begin{lstlisting}[firstnumber=0]
                   flag1:=F; flag2:=F; /*E\mrk{m:=F}E*/;
cr1:=F; wait1:=F;
do { // Process P1:
  flag1:=T; /*E\mrk{m:=F;}E*/
  while(flag2 & /*E\mrk{!m}E*/) {}
  cr1:=T;
  cr1:=F; flag1:=F; wait1:=T;
  while(F) //wait
    /*E\mrk{m:=F;}E*/
  wait1:=F; /*E\mrk{m:=F;}E*/
} while(T)
\end{lstlisting}
\end{minipage}
\hspace{0.1cm}
\begin{minipage}{0.54\textwidth}
\begin{lstlisting}[firstnumber=20]

cr2:=F; wait2:=F;
do { // Process P2:
  flag2:=T; /*E\mrk{m := !m;}E*/
  while(/*E\mrk{m}E*/) {} //wait to enter
  cr2:=T;
  cr2:=F; flag2:=F; wait2:=T; /*E\mrk{m:=F;}E*/
  while(F) //wait
    m:=F;
  wait2:=F; /*E\mrk{m:=F;}E*/
} while(T)
\end{lstlisting}
\end{minipage}
\end{figure}
The synthesis tool re-invents \texttt{turn}, but the solution is
complicated.  We had to construct a graph summarizing all runs to certify that
the specification holds.  This motivates our optimization feature, which can
be used to obtain simple solutions.

\textbf{Optimization.}
Next, we therefore add to the SMT formulation of the synthesis problem
constraints that count the updates of $m$ in an integer variable $c$, and also
add the constraint $c < \mathsf{Opt}$.  Now, we let the synthesis tool find a
minimum value for $\mathsf{Opt}$ such that the problem is still realizable.
Doing this, the tool will find an overly simplistic solution: it sets the
waiting condition \mrk{\ctrl$_{i,2}$} to $\true$, which means that the
synchronization needs to work only once.  When we consider the waiting 
condition to be an
input, we get the solution shown in Listing~\ref{lst:p1sol}, which has already
been discussed in Section~\ref{sec:ex}.

\textbf{Refinement.}
In Section~\ref{sec:ex}, we already discussed the refinement of the basic \ags 
solution with an additional variable \texttt{read}. Next, we refine the version 
with memory (Listing~\ref{lst:p1sol}) in the same way. By setting 
\mrk{\ctrl$_{2,3}$}$=f_{2,3}(\texttt{read})$, our tool finds the expected 
solution $f_{2,3}(\texttt{read}) = \neg \texttt{read}$ of toggling 
$\texttt{read}$ whenever the critical section is entered within $58$ seconds. % 
% Execute:
%   python3 ./ag_bosy.py --mono ../benchmarks/ag/pet_mem_in_replace
% Time:
%   58 seconds
% Solution:
%   (define-fun f2_2 ((x!1 (_ BitVec 1))) (_ BitVec 1)
%   (ite (= x!1 #b1) #b0
%   (ite (= x!1 #b0) #b1
%   #b0)))
Again, the modular refinement saved synthesis time.  Instead of $74+58 = 132$ 
seconds for synthesizing an \ags solution and refining it later, direct 
synthesis of the refined specification for both processes simultaneously 
requires $266$ seconds.

\subsection{Peer-To-Peer Filesharing}

\textbf{Sketch.} This example has been taken from~\cite{Fisman10} with slight 
modifications.  Two
\begin{wrapfigure}[6]{r}{0.52\textwidth}
%\vspace{-0.7cm}
\captionof{lstlisting}{Sketch of a filesharing protocol.}
\label{lst:peer0}
\vspace{-0.4cm}
\begin{minipage}{0.24\textwidth}
\begin{lstlisting}[firstnumber=0]
   u1:=T; d1:=F; u2:=F; d2:=F;
do{//process P1
 u1 := /*E\mrk{\ctrl$_{1,1}$}E*/
 d1 := u2 & /*E\mrk{\ctrl$_{1,2}$}E*/
} while(T)
\end{lstlisting}
\end{minipage}
\hspace{0.1cm}
\begin{minipage}{0.24\textwidth}
\begin{lstlisting}[firstnumber=20]

do{//process P2
 u2 := /*E\mrk{\ctrl$_{2,1}$}E*/
 d2 := u1 & /*E\mrk{\ctrl$_{2,2}$}E*/
} while(T)
\end{lstlisting}
\end{minipage}
\end{wrapfigure}
processes P1 and P2 use a peer-to-peer protocol to share files.  In each step,
process P$i$ can decide whether it wants to upload (by setting the variable
\texttt{u}$i$) or download (by setting \texttt{d}$i$).  A process can only 
download if the other
one uploads.  This is formalized by the sketch in Listing~\ref{lst:peer0}.

\textbf{Specification.}
Process $P1$ is specified by $(\always \eventually \texttt{d1}) \wedge (\always 
\eventually (\texttt{u1} \wedge \text{scheduled}(P1)))$ and $P2$ is specified by 
$(\always \eventually \texttt{d2}) \wedge (\always \eventually (\texttt{u2} 
\wedge \text{scheduled}(P2)))$.  The first conjunct of each specification 
expresses the goal of downloading infinitely often.  The second gives the other 
process the chance to do the same.

\textbf{Results.}
We set \mrk{\ctrl$_{1,j}$}$=f_{1,j}(\texttt{d2},\texttt{u2})$ and
\mrk{\ctrl$_{2,j}$}$=f_{2,j}(\texttt{d1},\texttt{u1})$, i.e.,  $P1$ makes its 
\begin{wrapfigure}[6]{r}{0.52\textwidth}
\vspace{-0.7cm}
\captionof{lstlisting}{Solution without AGS.}
\label{lst:peer1}
\vspace{-0.4cm}
\begin{minipage}{0.24\textwidth}
\begin{lstlisting}[firstnumber=0]
   u1:=T; d1:=F; u2:=F; d2:=F;
do{// process P1
 u1 := /*E\mrk{\texttt{(d2==u2)}}E*/;
 d1 := u2 & /*E\mrk{\texttt{!d2}}E*/
} while(T)
\end{lstlisting}
\end{minipage}
\hspace{0.1cm}
\begin{minipage}{0.24\textwidth}
\begin{lstlisting}[firstnumber=20]

do{// process P2
 u2 := /*E\mrk{\texttt{(u1==d1)}}E*/
 d2 := u1 & /*E\mrk{\texttt{!d1}}E*/
} while(T)
\end{lstlisting}
\end{minipage}
% This is a monolithic solution:
%    python3 ag_bosy.py --mono ../benchmarks/ag2/peer0
% says realizable
% This is not an AGS solution:
%    python3 ag_bosy.py ../benchmarks/ag2/peer0
% says unrealizable
\end{wrapfigure}
upload/download decisions based on the status of $P2$ and vice versa.  Without
\ags, we could\footnote{%
% BEGIN footnote
Without \ags, our tool could have produced this solution, but it actually
produces a different one.
% Execute:
%    python3 ag_bosy.py --mono ../benchmarks/ag2/peer
% Time:
%   0.4 seconds (1 realizability check)
% Solution:
% (define-fun f2_d ((x!1 (_ BitVec 2))) (_ BitVec 1)
% (ite (= x!1 #b11) #b0
% (ite (= x!1 #b10) #b1
% #b0)))
%
% (define-fun f2_u ((x!1 (_ BitVec 2))) (_ BitVec 1)
% (ite (= x!1 #b11) #b1
% (ite (= x!1 #b01) #b0
% (ite (= x!1 #b10) #b0
% (ite (= x!1 #b00) #b1
% #b0)))))
% (define-fun f1_d ((x!1 (_ BitVec 2))) (_ BitVec 1)
% (ite (= x!1 #b10) #b1
% (ite (= x!1 #b11) #b1
% #b1)))
% (define-fun f1_u ((x!1 (_ BitVec 2))) (_ BitVec 1)
% (ite (= x!1 #b10) #b0
% (ite (= x!1 #b01) #b0
% (ite (= x!1 #b11) #b0
% (ite (= x!1 #b00) #b1
% #b0)))))
The produced solution does not satisfy the \ags requirements either, but in a
way that is more difficult to explain.
% END footnote
} get the solution in Listing~\ref{lst:peer1}.
Figure~\ref{fig:peer_run_graph} summarizes all executions that are possible in
this implementation. Edges are labeled with scheduling decisions, and states
are labeled by the values of \texttt{u1}, \texttt{d1}, \texttt{u2}, and 
\texttt{d2} in this order.
\begin{figure}[htb]
  \begin{center}
    \includegraphics[width=0.8\textwidth]{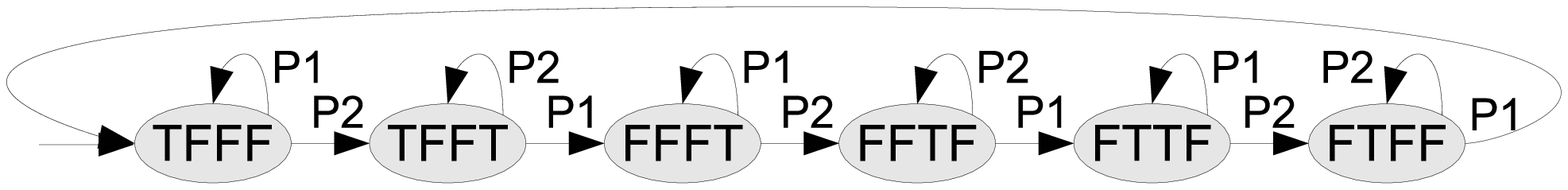}
    \caption{Run-graph summarizing all executions of Listing~\ref{lst:peer1}.}
    \label{fig:peer_run_graph}
  \end{center}
\end{figure}
Given that the
scheduler is fair, all depicted states will be visited infinitely often, so the
specification of both processes is fulfilled.  However, the correctness of this
solution depends on the fact that no process ever uploads and downloads
simultaneously.  If one process does, the other one gets stuck in a state
where it uploads but never downloads.  As a concrete example, consider an
alternative implementation of P2 with \mrk{\ctrl$_{2,j}$}$=\true$.  That is,
P2 always uploads and downloads at the same time.  The entire system will get
stuck in state TFTT, so the change in P2 makes P1 starve, although the
specification of P2 is still satisfied. Using our approach of \ags, we can be
sure that specification-preserving changes to one process cannot affect the
correctness of the other. Our tool computes an \ags solution within
one second for this example.
% Execute:
%    python3 ag_bosy.py ../benchmarks/ag/peer
% Time:
%   1.17 seconds (1 realizability check)
% Solution:
% (define-fun f1_d ((x!1 (_ BitVec 2))) (_ BitVec 1)
% (ite (= x!1 #b11) #b1
% (ite (= x!1 #b10) #b1
% #b1)))
% (define-fun f2_d ((x!1 (_ BitVec 2))) (_ BitVec 1)
% (ite (= x!1 #b10) #b1
% (ite (= x!1 #b11) #b1
% #b1)))
% (define-fun f2_u ((x!1 (_ BitVec 2))) (_ BitVec 1)
% (ite (= x!1 #b01) #b0
% (ite (= x!1 #b10) #b0
% (ite (= x!1 #b11) #b1
% (ite (= x!1 #b00) #b1
% #b0)))))
% (define-fun f1_u ((x!1 (_ BitVec 2))) (_ BitVec 1)
% (ite (= x!1 #b11) #b0
% (ite (= x!1 #b00) #b0
% (ite (= x!1 #b10) #b1
% (ite (= x!1 #b01) #b0
% #b0)))))

\else
More experiments can be found in Appendix~\ref{sec:app:exp}.
% peterson-example.tex will be included in the appendix
% p2p-example.tex will be included in the appendix
\fi

\subsection{Double-Buffering}

\smallskip
\noindent
\textbf{Sketch.} The example in Listing~\ref{lst:buf0} is taken 
from~\cite{VechevYY10} with slight adaptions.  It models a variant of the 
producer-consumer problem.  There are two buffers, \texttt{buf[0]} and 
\texttt{buf[1]}.  While process P1 writes to \texttt{buf[0]}, P2 reads from 
\texttt{buf[1]}.  Then, the buffers are swapped.  Such double-buffering is used 
in computer graphics and device drivers. We want to synthesize a rendezvous so 
that the two processes can never access the same buffer location simultaneously. 
 Hence, our (initial) specification for both processes is $\always(\neg 
\texttt{P1w} \vee \neg \texttt{P2r} \vee \texttt{fill} \neq \texttt{render} \vee 
i \neq j)$, where \texttt{P1w} indicates that P1 is in line~\ref{lst:buf0:w}, 
and \texttt{P2r} indicates that P2 is in line~\ref{lst:buf0:r}.

\begin{figure}[tb]
\captionof{lstlisting}{Sketch of a double buffering application.}
\label{lst:buf0}
\vspace{-0.4cm}
\begin{minipage}{0.44\textwidth}
\begin{lstlisting}[firstnumber=0]
                     fill:=T; render:=F;
i:=0; wait1:=F;
do { // process P1
 while(i < N) {
  buf[fill][i] := read();/*E\label{lst:buf0:w}E*/
  i := i + 1;
 }
 fill:=!fill; wait1:=T;
 while(/*E\mrk{\ctrl$_{1}$}E*/)//Sol.: /*E\mrk{\texttt{fill == render}}E*/
  { } // busy wait
 i:=0; wait1:=F;
} while(T)
\end{lstlisting}
\end{minipage}
\hspace{0.1cm}
\begin{minipage}{0.54\textwidth}
\begin{lstlisting}[firstnumber=20]

j:=0; wait2:=F;
do { // process P2
 while(j < N) {
  write(buf[render][j]) /*E\label{lst:buf0:r}E*/
  j := j + 1;
 }
 render:=!render; wait2:=T;
 while(/*E\mrk{\ctrl$_{2}$}E*/)//Sol.: /*E\mrk{\texttt{fill != render | !wait1}}E*/ /*E\label{lst:buf0:l}E*/
  { } // busy wait
 j:=0; wait2:=F;
} while(T)
\end{lstlisting}
\end{minipage}
\vspace{1cm}% To have the wrapfigure for the driver placed better
\end{figure}

%\noindent
\smallskip
\noindent
\textbf{Results.}
Our synthesis tool satisfies this specification with \mrk{\ctrl$_{i}$}$=\true$, 
so we add
the progress properties $\always\eventually(\texttt{P1w})$ and
$\always\eventually(\texttt{P2r})$ to get more meaningful solutions.
% Execute:
%   python3 ./ag_bosy.py ../benchmarks/ag/img1
% Time:
%   0.53 seconds
% Solution:
%   (define-fun f2_1 () (_ BitVec 1)
%   #b1)
%   (define-fun f1_1 () (_ BitVec 1)
%   #b1)
With
\mrk{\ctrl$_{i}$}$=f_i($\texttt{fill}, \texttt{render}$)$, the tool reports
unrealizability (without memory).
% Execute:
%   python3 ./ag_bosy.py ../benchmarks/ag/img2
% Time:
%   1.52 seconds
% Solution:
%   Unrealizable.
The solution of waiting while \texttt{fill}=\texttt{render} does not work 
because P2 could be stuck in line~\ref{lst:buf0:l} without being scheduled until 
P1 flips $\texttt{fill}$ again, which produces a deadlock.  But intuitively, 
there should exist a solution utilizing the equality 
\texttt{fill}=\texttt{render}.  Next, we therefore set
\mrk{\ctrl$_{1}$}$=f_1($\texttt{fill}=\texttt{render}, \texttt{wait2}$)$ and
\mrk{\ctrl$_{2}$}$=f_2($\texttt{fill}=\texttt{render}, \texttt{wait1}$)$.
This allows processes to observe whether the other one is waiting.  For $N=1$, 
we get the solution printed in comments in Listing~\ref{lst:buf0}. Essentially, 
the two processes take turns: by having opposite waiting conditions 
(\texttt{fill}=\texttt{render} vs. \texttt{fill}$\neq$\texttt{render}) one waits 
while the other works.  The additional disjunct \texttt{!wait1} in P2 is only 
useful in the first iteration: if P2 finishes first, it waits although 
\texttt{fill}$\neq$\texttt{render}.

% \begin{figure*}
% \captionof{lstlisting}{Solution for the sketch in Listing~\ref{lst:buf0}.}
% \label{lst:buf1}
% \vspace{-0.4cm}
% \begin{minipage}{0.49\textwidth}
% \begin{lstlisting}[firstnumber=0]
%                                      fill:=true; render:=false;
% i:=0; wait1:=false;
% do { // process P1
%   while(i < N) {
%     buf[fill][i] := read(); 
%     i := i + 1;
%   }
%   fill:=!fill; wait1:=true;
%   while(/*E\mrk{\texttt{fill == render}}E*/)
%     { } // busy wait
%   i:=0; wait1:=false;
% } while(true)
% \end{lstlisting}
% \end{minipage}
% \hspace{0.1cm}
% \begin{minipage}{0.49\textwidth}
% \begin{lstlisting}[firstnumber=20]
% 
% j:=0; wait2:=false;
% do { // process P2
%   while(j < N) {
%     write(buf[render][j])
%     j := j + 1;
%   }
%   render:=!render; wait2:=true;
%   while(/*E\mrk{\texttt{fill != render | !wait1}}E*/)
%     { } // busy wait
%   j:=0; wait2:=false;
% } while(true)
% \end{lstlisting}
% \end{minipage}
% % Execute:
% %   python3 ./ag_bosy.py ../benchmarks/ag/buf/buf1
% % Time:
% %   0.97 seconds
% % (define-fun f2_1 ((x!1 (_ BitVec 2))) (_ BitVec 1)
% % (ite (= x!1 #b11) #b0
% % (ite (= x!1 #b10) #b1
% % (ite (= x!1 #b00) #b1
% % (ite (= x!1 #b01) #b1
% % #b1)))))
% % (define-fun f1_1 ((x!1 (_ BitVec 2))) (_ BitVec 1)
% % (ite (= x!1 #b11) #b1
% % (ite (= x!1 #b00) #b0
% % (ite (= x!1 #b10) #b0
% % (ite (= x!1 #b01) #b1
% % #b0)))))
% \end{figure*}

\smallskip
\noindent
\textbf{Performance.} 
Table~\ref{tab:buf_time} lists the synthesis times for resolving the sketch of
\begin{wraptable}[6]{r}{0.52\textwidth}
\vspace{-0.6cm}
\setlength{\tabcolsep}{2.5pt}
\caption{Synthesis times [sec] for Listing~\ref{lst:buf0}.}
\vspace{-0.3cm}
\label{tab:buf_time}
\begin{tabular}{|l|c|c|c|c|c|c|c|c|c|}
\hline
N       & 1 & 2 & 3 & 4 & 5 & 6 & 7 & 8  & 15 \\
\hline
\ags     & 1 & 5 & 5 & 54& 51& 49& 47&1097& 877\\
non-\ags & 1 & 4 & 4 & 38& 35& 32& 31& 636& 447\\
\hline
%AGS 16: 13874.99s
\end{tabular}
% \setlength{\tabcolsep}{8pt}
% \begin{tabular}{lcc}
% \toprule
% N  & \ags  & non-\ags  \\
% \midrule
% 1  & 1     & 1         \\
% 2  & 5     & 4         \\
% 3  & 5     & 4         \\
% 4  & 54    & 38        \\
% 5  & 51    & 35        \\
% 6  & 49    & 32        \\
% 7  & 47    & 31        \\
% 8  & 1097  & 636       \\
% 15 &  887  & 447       \\
% \bottomrule
% \end{tabular}
\end{wraptable}
Listing~\ref{lst:buf0} with 
increasing $N$.  We use bitvectors for encoding
the counters $i$ and $j$, and observe that the computation time mostly depends
on the bit-width. This explains the jumps whenever $N$ reaches the next
power of two.  Cooperative co-synthesis is only slightly faster than \ags on
this example.
% (where we do not specifically exploit \ags features).

\subsection{Synthesis of Atomic Sections in a Driver}

\textbf{Program.}
This example is taken from~\cite{CernyHRRT13} (called \texttt{ex5} there),
and is a simplified 
version of a bug in the \textsf{i2c} driver of the Linux
kernel\footnote{See
\url{
http://kernel.opensuse.org/cgit/kernel/commit/?id=7a7d6d9c5fcd4b674da38e814cfc07
24c67731b2}}.
The code is shown in Listing~\ref{lst:driver0}.
\begin{wrapfigure}[8]{r}{0.52\textwidth}
\vspace{-0.7cm}
\captionof{lstlisting}{Bug in i2c driver (simplified).}
\label{lst:driver0}
\vspace{-0.4cm}
\begin{minipage}{0.25\textwidth}
\begin{lstlisting}[firstnumber=0]
          Open:=0; On:=F;
do{//process P1
 if(Open < MAX){
  if(Open == 0)
   On := T;/*E\label{lst:p1:start}E*/
  Open++;/*E\label{lst:p1:end}E*/
 }
} while(T)
\end{lstlisting}
\end{minipage}
\hspace{0.1cm}
\begin{minipage}{0.24\textwidth}
\begin{lstlisting}[firstnumber=20]

do{//process P2
 if(Open > 0){/*E\label{lst:p2:loop}E*/
  Open--;
  if(Open == 0)/*E\label{lst:p2:start}E*/
   On := F;/*E\label{lst:p2:end}E*/
 }
} while(T)
\end{lstlisting}
\end{minipage}
\end{wrapfigure}
Process P1 opens sessions and
P2 closes them.  The variable \texttt{Open} counts the currently opened
sessions. If there are open sessions, \texttt{On} is set to $\true$,
otherwise to $\false$. Due to a race condition, it can happen
that $\texttt{Open} \neq 0$, but $\texttt{On} = \false$.\footnote{by executing
the lines 2, 3, 4, 5, 22, 23, 24, 2, 3, 4, 5, 25 in a row.}  We will
now use our engine to synthesize atomic sections so that this problem cannot
occur.

\smallskip
\noindent
\textbf{Modeling.}
We search for two functions $f_1$ and $f_2$ that map the program counter value 
of the respective process to $\true$ or $\false$.  If a program counter value is 
mapped to $\true$, then this means that the process cannot be interrupted at 
this point in the program, but immediately continues to execute the next 
instruction.  That is, the two adjacent instructions are executed atomically. 
(Each line is considered an instruction.) In the SMT input file to our synthesis 
tool, this is modeled by making process P$1$ do nothing if it is scheduled but 
$f_2(pc2)$ is $\true$, and vice versa.  This way, we do not have to change the 
scheduler to take atomic sections into account, but rather ignore ``wrong'' 
scheduling decisions in our transition relation, which has the same effect under 
a fair scheduler.

\smallskip
\noindent
\textbf{Specification.}
Using $\spec = \always((\texttt{Open} = 0) \vee \texttt{On})$ as the sole 
specification for both processes is not ideal: one process could make the other 
starve by building an atomic loop (e.g, by making all statements atomic). 
This enforces the specification, but is not desirable.
% Execute:
%   python3 ./ag_bosy.py --opt_from 8 --opt_to 1 ../benchmarks/ag/driver0
% Time:
%   27 seconds (8 realizability checks)
% Solution:
%   (define-fun f1_1 ((x!1 (_ BitVec 2))) (_ BitVec 1)
%   (ite (= x!1 #b00) #b0
%   (ite (= x!1 #b01) #b0
%   (ite (= x!1 #b10) #b0
%   (ite (= x!1 #b11) #b0
%   #b0)))))
%   (define-fun f2_1 ((x!1 (_ BitVec 2))) (_ BitVec 1)
%   (ite (= x!1 #b00) #b1
%   (ite (= x!1 #b01) #b0
%   (ite (= x!1 #b10) #b0
%   (ite (= x!1 #b11) #b0
%   #b0)))))
Hence, we specify process P1 by $\spec \wedge \always(\eventually( 
\text{scheduled}(P2) \wedge \neg f_1(pc1)))$ and P2 by $\spec \wedge 
\always(\eventually( \text{scheduled}(P1) \wedge \neg f_2(pc2)))$.  This way, 
both processes allow the other one to move infinitely often.  

\smallskip
\noindent
\textbf{Results.} For performance reasons, we prefer solutions with a low number 
of atomic sections.  Hence, we assign costs to active atomic sections, and 
let our tool minimize the total costs.  As a result, we get an atomic 
section between line~\ref{lst:p1:start} and~\ref{lst:p1:end}, and another one 
between line~\ref{lst:p2:start} and~\ref{lst:p2:end}. This renders all updates 
of the variable \texttt{On} atomic with the relevant accesses of 
\texttt{Open}, and thus fixes the race condition.  Both \ags and cooperative 
co-synthesis produce the same solution for this example within 54 and 35 
seconds.
% Execute:
%   python3 ./ag_bosy.py --mono --opt_from 8 --opt_to 1 ../benchmarks/ag/driver
% Time:
%   35 seconds (7 realizability checks)
% Solution:
% (define-fun f1_1 ((x!1 (_ BitVec 2))) (_ BitVec 1)
% (ite (= x!1 #b00) #b0
% (ite (= x!1 #b01) #b0
% (ite (= x!1 #b10) #b0
% (ite (= x!1 #b11) #b1
% #b0)))))
% (define-fun f2_1 ((x!1 (_ BitVec 2))) (_ BitVec 1)
% (ite (= x!1 #b00) #b0
% (ite (= x!1 #b01) #b0
% (ite (= x!1 #b10) #b0
% (ite (= x!1 #b11) #b1
% #b0)))))
% )
%
% Execute:
%   python3 ./ag_bosy.py --opt_from 8 --opt_to 1 ../benchmarks/ag/driver
% Time:
%   54 seconds (7 realizability checks)
% Solution:
% (define-fun f2_1 ((x!1 (_ BitVec 2))) (_ BitVec 1)
% (ite (= x!1 #b00) #b0
% (ite (= x!1 #b01) #b0
% (ite (= x!1 #b10) #b0
% (ite (= x!1 #b11) #b1
% #b0)))))
% (define-fun f1_1 ((x!1 (_ BitVec 2))) (_ BitVec 1)
% (ite (= x!1 #b00) #b0
% (ite (= x!1 #b01) #b0
% (ite (= x!1 #b10) #b0
% (ite (= x!1 #b11) #b1
% #b0)))))

\section{Related Work}

\textbf{Reactive synthesis.}
Automatic synthesis of reactive programs from formal 
specifications, as defined by Church~\cite{Church62}, is usually reduced 
either 
to games on finite graphs~\cite{BL69}, or to the 
emptiness problem of automata over infinite trees~\cite{Rabin72}. Pnueli and 
Rosner~\cite{PnueliR89} proposed synthesis from LTL specifications, and 
showed its 2EXPTIME complexity based on a doubly exponential
translation of the specification into a tree automaton. We use 
extensions of the game-based approach (see below) to obtain new complexity 
results for \ags, while our implementation uses an encoding based on 
tree automata~\cite{FinkbeinerS13} that avoids one exponential blowup compared 
to the standard approaches~\cite{KupfermanV05}. 

We consider the synthesis of concurrent or distributed reactive systems with 
partial information, which has been shown to be undecidable in 
general~\cite{PnueliR90}, even 
for simple safety fragments of temporal logics~\cite{Schewe14}. 
Several approaches for distributed synthesis have been proposed, either by 
restricting the specifications to be local to each 
process~\cite{MadhusudanT01}, by restricting the communication graph to 
pipelines and similar structures~\cite{FinkbeinerS05}, or by falling back to 
semi-decision procedures that will 
eventually find an implementation if one exists, but in
general cannot detect unrealizability of a specification~\cite{FinkbeinerS13}.
Our synthesis 
approach is based on the latter, and extends it with synthesis from program 
sketches~\cite{Solar-Lezama13}, as well as the assume-guarantee 
paradigm~\cite{ChatterjeeH07}.
%
%In the literature, bounded synthesis has only been described for the 
%synthesis of components from scratch. For the 
%distributed case, the notion of \emph{partial designs}~\cite{ScheweF06} comes 
%close to using program sketches. However, in partial designs it is assumed 
%that some of the processes are completely implemented (white boxes), while 
%others are completely free (black boxes).

\medskip
\noindent
\textbf{Graph games.} Graph games provide a mathematical foundation
to study the reactive synthesis problem~\cite{Church62,BL69,GH82}.
For the traditional perfect-information setting, the complexity of solving 
games has been deeply studied; e.g., for reachability and safety objectives 
the problem is PTIME-complete~\cite{Immerman81,Beeri};
for GR(1) the problem can be solved in polynomial time~\cite{PPS06};
and for LTL the problem is 2EXPTIME-complete~\cite{PnueliR89}.
For two player partial-information games with reachability 
objectives, EXPTIME-completeness was established in~\cite{Reif84}, and 
symbolic algorithms and strategy construction procedures were studied 
in~\cite{CDHR07,BCDDH10}.
However, in the setting of multi-player partial-observation games, the 
problem is undecidable even for three players~\cite{Peterson79} and for 
safety objectives as well~\cite{CHOP13}.
While most of the previous work considers only the general problem and its
complexity, the complexity distinction we study for memoryless strategies, 
and the practical 
SMT-based approach to solve these games has not been studied before.

\medskip
\noindent
\textbf{Various equilibria notions in games.}
In the setting of two-player games for reactive synthesis, the goals of the 
two players are complementary (i.e., games are zero-sum). 
For multi-player games there are various notions of equilibria studied for 
graph games, such as Nash equilibria~\cite{Nash50} for graph games
%~\cite{CMJ04} 
that inspired notions of  rational synthesis~\cite{Fisman10}; refinements of 
Nash equilibria such as secure equilibria~\cite{CHJ06} 
that inspired assume-guarantee synthesis (AGS)~\cite{ChatterjeeH07}, and 
doomsday equilibria~\cite{CDFR14}. An alternative to Nash equilibria and its 
refinements are approaches based on iterated admissibility~\cite{BrenguierRS14}.
Among the various equilibria and synthesis notions, the most relevant one for 
reactive synthesis is 
AGS, which is applicable for synthesis of mutual-exclusion
protocols~\cite{ChatterjeeH07} as well as for security protocols~\cite{CR14}.
The previous work on AGS is severely restricted by perfect information, 
whereas we consider
the problem under the more general framework of partial-information (the need
of which was already advocated in applications in~\cite{JMM11}).

\ifextended
\medskip
\noindent
\textbf{Synthesis of program fragments, sketching.}
For functional programs, where the specification is a relation between a 
single pair of inputs and outputs that can be represented as a first-order 
logic formula, early works~\cite{WaldingerL69,Green69,MannaW71} were based on
extensions of first-order theorem provers with induction and proof analysis. 
Recent methods leverage the power 
of decision procedures to obtain completeness even when reasoning about 
infinite data types~\cite{KuncakMPS13}, as well as techniques 
that limit the control structure of the synthesized program by bounding the 
resources of the program~\cite{SrivastavaGF13}. 

A special form of the latter 
approach is \emph{program sketching}~\cite{Solar-LezamaTBSS06,Solar-Lezama13},
where the control 
structure of the program is given, and values for a fixed number of 
unknown variables are determined by search techniques. Our approach is 
inspired by program sketching, in that we use sketches to limit the control 
structure of synthesized programs. We go beyond standard program sketching in that our 
programs are reactive, and in general can use an unbounded amount of memory 
in addition to the program variables in the sketch. Moreover, the search for 
suitable valuations of variables in sketching is usually implemented as a 
counterexample-guided inductive synthesis (CEGIS) loop, whereas we use an 
automata-based approach that encodes the existence of a solution 
into a single SMT problem. In synthesis of reactive systems, 
synthesis from \emph{partial designs} allows to start with a distributed 
system 
where some components are already implemented~\cite{FinkbeinerS05,FinkbeinerS13}, 
and an approach similar to CEGIS has been proposed as 
\emph{lazy synthesis}~\cite{FinkbeinerJ12}.
\else
\fi

\section{Conclusion}
\label{sec:concl}

Assume-Guarantee Synthesis (\ags) is particularly suitable for
concurrent reactive systems, because none of the synthesized processes relies on the concrete realization of the others. This feature makes a synthesized
solution robust against changes in single processes.  A major limitation of
previous work on \ags was that it assumed perfect information about all
processes, which implies that synthesized implementations may use local
variables of other processes.  In this paper, we resolved this shortcoming by
(1) defining \ags in a partial information setting, (2) proving new complexity
results for various sub-classes of the problem, (3) presenting a pragmatic
synthesis algorithm based on the existing notion of bounded synthesis to solve
the problem, (4) providing the first implementation of \ags, which also supports
the optimization of solutions with respect to user-defined cost functions, and
(5) demonstrating its usefulness by resolving sketches of several concurrent
protocols. We believe our contributions can form an important step towards a
mixed imperative/declarative programming paradigm for concurrent programs, where
the user writes sequential code and the concurrency aspects are taken care of
automatically.

In the future, we plan to work on issues such as scalability
and usability of our prototype, explore applications for security protocols
as mentioned in~\cite{JMM11}, and research restricted cases where the \ags
problem with partial information is decidable.

\vspace{-1em}
{\scriptsize
\bibliographystyle{plain}
\bibliography{paper}
}
\ifextended
\else
\newpage
\appendix
\section{Complexity and Decidability of \ags}
\label{sec:complexity}

\subsection{Complexity Results}

\newpage
\section{Correctness of the \ags Algorithm}
\label{correctness-appendix}

\PropOne*

\begin{proof}
From our definitions of refinement and of the transition functions $\trans_i'$,
it is obvious that a model will represent a refinement of the given program 
sketches.

Furthermore, by correctness of the annotation approach from 
bounded synthesis~\cite{FinkbeinerS13}, any transition function that 
satisfies the constraints will satisfy the specification (and the combination 
of $\trans_i'$ is in particular a transition function).
\end{proof}

\PropTwo*

\begin{proof}
Suppose such $P_1', P_2'$ exist. By the definition of refinement, we have 
that for all $\state \in \valsetof{\statevarset_i'}
, \obsval \in \valsetof{\obsvarset_i}, \inputval \in 
\valsetof{\inputvarset'_i}$ there exists $\controlval \in 
\valsetof{\controlvarset_i}$ with
\[ \trans_i'(\state,\obsval,\inputval)\restrict{\statevarset_i} = 
\trans_i(\state\restrict{\statevarset_i}, \obsval, \inputval \ccat \controlval
). \]

The control valuations $\controlval$ for different valuations $\state,\obsval,
\inputval$ of the other variables give us a 
model of the function $f_i$ that computes the controllable variables. In a 
similar way, the computation of memory valuations for different valuations of 
the other variables gives us a model of the function $\memtrans_i$.
\end{proof}

\newpage

\newpage
\section{More Experiments} \label{sec:app:exp}

\fi

\end{document}